\documentclass[12pt]{article}

\newcommand{\fmsdonly}[1]{} %
\newcommand{\ieeeonly}[1]{}  %
\newcommand{\artonly}[1]{#1}  %
\newcommand{\artfmsdonly}[1]{#1}  %

\newcommand{\fullonly}[1]{#1} 
\newcommand{\shortonly}[1]{} 
\newcommand{\journal}[1]{#1} %
\newcommand{\thesis}[1]{}
\newcommand{\unused}[1]{}

\ieeeonly{\usepackage{cite}}
\usepackage{graphicx}
\usepackage{bbm}
\usepackage{upgreek}
\usepackage{xspace}

\usepackage{amsmath,amssymb,amsfonts}
\usepackage{amsthm}
\usepackage{manyfoot}
\usepackage{booktabs}
\usepackage{algorithm}
\usepackage{algorithmicx}
\usepackage{algpseudocode}
\artonly{\usepackage[margin=1in]{geometry}}
\artonly{\usepackage{authblk}}
\artonly{\usepackage{hyperref}}
\ieeeonly{\usepackage{hyperref}}
\ieeeonly{\usepackage{textcomp}}
\fmsdonly{\sloppy}

\ieeeonly{
\usepackage{bm}
\makeatletter
\AtBeginDocument{\DeclareMathVersion{bold}
\SetSymbolFont{operators}{bold}{T1}{times}{b}{n}
\SetSymbolFont{NewLetters}{bold}{T1}{times}{b}{it}
\SetMathAlphabet{\mathrm}{bold}{T1}{times}{b}{n}
\SetMathAlphabet{\mathit}{bold}{T1}{times}{b}{it}
\SetMathAlphabet{\mathbf}{bold}{T1}{times}{b}{n}
\SetMathAlphabet{\mathtt}{bold}{OT1}{pcr}{b}{n}
\SetSymbolFont{symbols}{bold}{OMS}{cmsy}{b}{n}
\renewcommand\boldmath{\@nomath\boldmath\mathversion{bold}}}
\makeatother

\def\BibTeX{{\rm B\kern-.05em{\sc i\kern-.025em b}\kern-.08em
    T\kern-.1667em\lower.7ex\hbox{E}\kern-.125emX}}
}

\newtheorem{theorem}{Theorem}
\newtheorem{lemma}{Lemma}%

\newtheorem{definition}{Definition}%

\newcommand{\name}{Bb-Simplex\xspace}

\newcommand{\delete}[1]{}

\begin{document}

\newcommand{\acks}{This work was supported in part by National Science Foundation awards 
ITE-2134840, ITE-2040599, %
CCF-1954837, %
CCF-1918225, %
and CPS-1446832. %
}

\ieeeonly{\doi{TBD}}

\shortonly{
\title{A Barrier Certificate-based Simplex Architecture for Systems with Approximate and Hybrid Dynamics}
}

\fullonly{
\title{A Barrier Certificate-based Simplex Architecture for Systems with Approximate and Hybrid Dynamics}
}

\fmsdonly{
\author[1]{\fnm{Amol} \sur{Damare}}\email{adamare@cs.stonybrook.edu}
\author[1]{\fnm{Shouvik} \sur{Roy}}\email{sroy@cs.stonybrook.com}
\author[2]{\fnm{Roshan} \sur{Sharma}}\email{roshan.sharma@comed.com}
\author[2]{\fnm{Keith} \sur{DSouza}}\email{keith.dsouza@comed.com}
\author[1]{\fnm{Scott A.} \sur{Smolka}}\email{sas@cs.stonybrook.com}
\author*[1]{\fnm{Scott D.} \sur{Stoller}}\email{stoller@cs.stonybrook.edu}

\affil[1]{\orgdiv{Department of Computer Science}, \orgname{Stony Brook University}, \orgaddress{%
\city{Stony Brook}, %
\state{NY}, \country{USA}}}

\affil[2]{\orgdiv{Smart Grid---Emerging Technology}, \orgname{Commonwealth Edison}, \orgaddress{%
\city{Oakbrook Terrace}, %
\state{IL}, \country{USA}}}
}

\artonly{
\author[1]{Amol Damare}
\author[1]{Shouvik Roy}
\author[2]{Roshan Sharma}
\author[2]{Keith DSouza}
\author[1]{Scott A.\ Smolka}
\author[1]{Scott D.\ Stoller}
\affil[1]{Department of Computer Science, Stony Brook University}
\affil[2]{Smart Grid---Emerging Technology Commonwealth Edison}
}

\ieeeonly{
\author{\uppercase{Amol Damare}\authorrefmark{1},
\uppercase{Shouvik Roy}\authorrefmark{1},
\uppercase{Roshan Sharma}\authorrefmark{2},
\uppercase{Keith DSouza}\authorrefmark{2},
\uppercase{Scott A.\ Smolka}\authorrefmark{1}, and
\uppercase{Scott D.\ Stoller}\authorrefmark{1}
\address[1]{Department of Computer Science, Stony Brook University}
\address[2]{Smart Grid---Emerging Technology Commonwealth Edison}
}
\tfootnote{\acks}
\corresp{Corresponding author: Scott D.\ Stoller (e-mail: stoller@cs.stonybrook.edu).}
}

\newcommand{\abstracttext}{We present \emph{Barrier-based Simplex}
(\name), a new, provably correct design for runtime assurance of continuous dynamical systems. \name is centered around the Simplex control architecture, which consists of a high-performance \emph{advanced controller} that is not guaranteed to maintain safety of the plant, a verified-safe \emph{baseline controller}, and a \emph{decision module} that switches control of the plant between the two controllers to ensure safety without sacrificing performance.
In \name, \emph{Barrier certificates} are used to prove that the baseline controller ensures safety.  Furthermore, \name features a new automated method for deriving, from the barrier certificate, the conditions for switching between the controllers.  Our method is based on the Taylor expansion of the barrier certificate and yields computationally inexpensive switching conditions. 

We also propose extensions to \name to enable its use in \emph{hybrid systems}, which have multiple modes each with its own dynamics, and to support its use when only \emph{approximate dynamics} (not exact dynamics) are available, for both continuous-time and hybrid dynamical systems. 

We consider significant applications of \name to microgrids featuring advanced controllers in the form of neural networks trained using reinforcement learning.  These microgrids are modeled in RTDS, an industry-standard high-fidelity, real-time power systems simulator.  Our results demonstrate that \name can automatically derive switching conditions for complex continuous-time and hybrid systems, the switching conditions are not overly conservative, and \name ensures safety even in the presence of adversarial attacks on the neural controller when only approximate dynamics (with an error bound) are available.
}

\fmsdonly{\abstract{\abstracttext}}
\fmsdonly{}

\ieeeonly{
\begin{abstract}
\abstracttext
\end{abstract}
\titlepgskip=-21pt
}

\maketitle

\artonly{
\begin{abstract}
\abstracttext
\end{abstract}}

\section{Introduction}      
\label{sec:intro}
\label{ch:Intro}
\emph{Barrier certificates} (BaCs)~\cite{prajna2004,prajna2006} are a powerful method for verifying the safety of continuous dynamical systems without explicitly computing the set of reachable states.  A BaC is a function of the state satisfying a set of inequalities on the value of the function and value of its time derivative along the dynamic flows of the system.  Intuitively, the zero-level-set of a BaC forms a ``barrier'' between the reachable states and unsafe states. Existence of a BaC assures that starting from a state where the BaC is positive, safety is forever maintained~\cite{magnus2015,prajna2006,prajna2004}. Moreover, there are automated methods to synthesize BaCs, e.g.,~\cite{Kundu2019,permissiveBC,Zhao2020,Meng2021}.

Proving safety of plants with complex controllers is difficult with any formal verification technique, including barrier certificates.  As we now show, however, BaCs can play a crucial role in applying the well-established Simplex Control Architecture~\cite{Sha1998,LSha2001} to provide provably correct runtime safety assurance for systems with complex controllers.

We present \emph{Barrier-based Simplex}
(\name), a new, provably correct design for runtime assurance of continuous dynamical systems (see~\cite{Damare2022} for an earlier conference version of this work). \name is centered around the Simplex Control Architecture, which consists of a high-performance \emph{advanced controller} (AC) that is not guaranteed to maintain safety of the plant, a verified-safe \emph{baseline controller} (BC), and a \emph{decision module} that switches control of the plant between the two controllers to ensure safety without sacrificing performance.
In \name, \emph{Barrier certificates} are used to prove that the baseline controller ensures safety.  Furthermore, \name features a new scalable (relative to existing methods that require reachability analysis, e.g., \cite{bak2010,bak2011,johnson2016_rsl}) and automated method for deriving, from the BaC, the conditions for switching between the controllers. Our method is based on the Taylor expansion of the BaC and yields computationally inexpensive
switching conditions.

We also present multiple extensions to \name. First, we extend \name to allow the use of hybrid systems. We propose and prove correctness of the switching condition based on checking the Taylor approximation of the BaC for each reachable mode of the hybrid system in the next control period.  Second, we extend \name to support the use of approximate dynamics in the underlying system. This extension is especially helpful when the system is too complex to model accurately or an analytical model is unavailable. We bound the errors of approximate dynamics using \emph{approximate trace conformance}(ATC), and then modify the definition of a BaC to include these bounds. We amend the training and verification methods proposed in~\cite{Zhao2021} according to the modified defintion of BaCs. With error bounds and modified BaCs, we propose a new switching condition for approximate dynamics and prove its correctness.  Finally, we combine these two extensions of \name (hybrid systems and approximate dynamics) to obtain an end-to-end runtime assurance framework that works for both continuous and hybrid systems, and supports the use of approximate dynamics.

We consider a significant application of \name, namely \emph{microgrid control}. A \emph{microgrid} is an integrated energy system comprising distributed energy resources and multiple energy loads operating as a single controllable entity in parallel to, or islanded from, the existing power grid~\cite{ton2012}.  The microgrid we consider features an advanced controller (for voltage control) in the form of a neural network trained using reinforcement learning. For this purpose, we use \name in conjunction with the \emph{Neural Simplex Architecture} (NSA)~\cite{dung2019_nsa}, where the AC is an AI-based \emph{neural controller} (NC). NSA also includes an \emph{adaptation module} (AM) for online retraining of the NC while the BC is in control.

The microgrids we consider are modeled in RTDS~\cite{rtds}, an industry-standard high-fidelity, real-time power systems simulator.  Our results demonstrate that \name can automatically derive switching conditions for complex systems, the switching conditions are not overly conservative, and \name ensures safety even in the presence of adversarial attacks on the neural controller.

\paragraph*{Architectural overview of \name}
Figure~\ref{fig:archi} shows the overall architecture of the
combined Barrier-based
Neural Simplex Architecture. The green part of the figure depicts our design methodology; the blue part illustrates NSA.
Given the BC, the required safety properties, and a dynamic model of the plant,
our methodology generates a BaC and then derives the switching condition from it.  The reinforcement learning module learns a high-performance NC based on the performance objectives encoded in the reward function.

This paper is an extended version of~\cite{Damare2022} with significant new material including the extension to hybrid systems, the extension to approximate dynamics, the combination of these extensions, and an all-new series of experiments evaluating \name and all its extensions using a much larger and more complex model of a real MG.

\thesis{
The structure of the rest of the paper is the following.  Section~\ref{sec:BaC} provides background material on barrier certificates.  Section~\ref{sec:DM} presents \name for continuous dynamical systems, including our new approach for deriving switching conditions from barrier certificates.   Sections~\ref{ch:ApprxDyn}, \ref{sec:hybrid}, and \ref{sec:approxHybrid} present the extensions to approximate dynamics, hybrid systems, and their combination, respectively.    
 Section~\ref{sec:EE_BCM} presents the results of applying \name to a realistic model of a microgrid that is operational in the Bronzeville neighbourhood of Chicago.  Section~\ref{sec:related} discusses related work. Section~\ref{sec:conclusion} offers concluding remarks.
}

\begin{figure}[t]
\centering
\includegraphics[width=0.99\columnwidth]{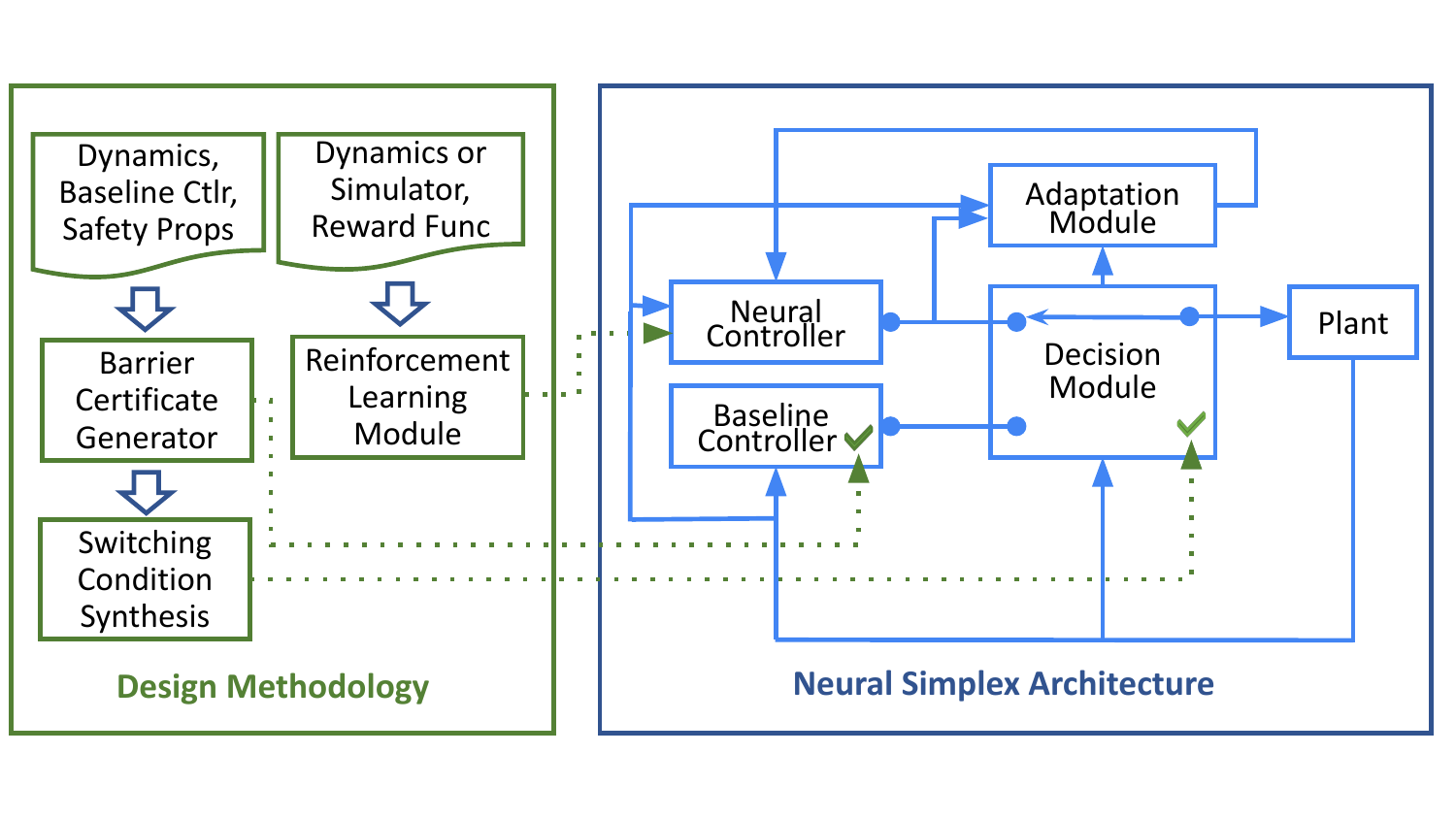}\hfill\vspace{-2.0ex}
\caption{Overview of the Barrier Certificate-based Neural Simplex Architecture}
\vspace*{-2ex}
\label{fig:archi}
\end{figure}

\section{Preliminaries}      
\label{sec:BaC}
\label{ch:BaC}

We use Barrier Certificates (BaCs) to prove that the BC ensures safety.  Also, our design methodology for computing switching conditions (see Section~\ref{sec:DM}) requires a BaC, but is independent of how the BaC is obtained.  We implemented two automated methods for BaC synthesis from the literature.  One method is based on sum-of-squares optimization (SOS)~\cite{Kundu2019,permissiveBC}.  Its applicability is limited to systems with polynomial dynamics.

The other method is \emph{SyntheBC}~\cite{Zhao2021}, which uses deep learning to synthesize a BaC. First, training samples obtained by sampling different areas  of the state space are used to train a feedforward ReLU neural network with two hidden layers as a candidate BaC.  Second, the validity of this candidate BaC must be verified.  The NN's structure allows the problem of checking whether the NN satisfies the defining conditions of a BaC to be transformed into mixed-integer linear programming (MILP) and mixed-integer quadratically-constrained programming (MIQCP) problems, which we solve using the Gurobi optimizer.  If the verification fails, the Gurobi optimizer provides counter-examples which can be used to guide retraining of the NN.  In this way, the training and verification steps can be iterated as needed.

\section{\name for Continuous-Time Dynamics}
\label{sec:DM}
\label{ch:DM}
\newcommand{\recov}{\mathcal{R}}
\newcommand{\zsls}{\mathcal{Z}}

\newcommand{\admis}{\mathcal{A}}

This \journal{Section}\thesis{Chapter} presents our novel methodology for deriving the switching logic from the BC's BaC for systems with continuous dynamics.
The Decision Module (DM) implements this switching logic for forward and reverse switching. When the Forward Switching Condition (FSC) is true, control is switched from the NC to the BC; likewise, when the reverse-switching condition (RSC) is true, control is switched from the BC to the NC. 
Consider a continuous dynamical system of the form:
\begin{equation}
\label{eqn:dynamics}
\dot{x} =f(x, g(x))
\end{equation}
where $x \in \mathbb{R}^k$ is the state of the system, $f$ is the dynamics of the physical system (a.k.a.\  the plant), and real-valued function $g$ is the control law of the controller.  The arguments of $f$ are a state $x$ and a control input $u$.  The controller $g$ produces a control action, given the current state $x$.  In \name, $g$ is instantiated with the BC's control law when deriving a BaC.  The set of all valid control actions 
is denoted by $\mathbb{U}$. The set of \emph{unsafe states} is denoted by $\mathcal{U}$. The set of \emph{initial states} is denoted by $\mathbb{I}$. Let $x_{lb}, x_{ub} \in \mathbb{R}^k$ be \emph{operational bounds} on the ranges of state variables, reflecting physical limits and simple safety requirements.

The set $\admis$ of \emph{admissible states} is given by: $\admis=\{ x:  x_{lb} \leq x \leq x_{ub} \}$. A state of the plant is \emph{recoverable} if the BC can take over in that state and keep the plant invariably safe. For a given BC, we denote the \emph{recoverable region} by $\recov$. Note that $\mathcal{U}$ and $\recov$ are disjoint. 
The safety of such a system under control of the BC can be established using a BaC $h(x):\mathbb{R}^k \rightarrow \mathbb{R}$ of the following form~\cite{prajna2004,prajna2006,permissiveBC,Kundu2019}:
\begin{definition}
\label{def:BaC}
\begin{equation}
\label{eq:bac}
    \begin{aligned}
        h(x) \leq 0, \quad &\forall x \in \mathbb{I} \\
        h(x) > 0, \quad &\forall x \in \mathcal{U}\\
        (\nabla_x h)^T f(x, g(x)) \leq 0,  \quad &\forall x \in \mathbb{R}^k ~\text{s.t.}~ h(x)=0
    \end{aligned}
\end{equation}
\end{definition}
The BaC is positive over the unsafe region and non-positive otherwise. $\nabla_x h$ is the gradient of $h$ w.r.t.\ $x$ and the expression $(\nabla_x h)^T f(x,g(x))$ is the time derivative of $h$. The zero-sub-level set of a BaC $h$ is $\zsls(h) =\{ x : h(x) < 0\}$. In~\cite{permissiveBC}, the invariance of this set is used to show $\zsls(h) \subseteq \recov$.

An FSC is a function of the current state $x$ and the control action $u$ output by the AC in that state. Let $\eta$ denote the control period a.k.a.\ time step. Let $\hat{h}(x,u,\delta)$ denote the degree-$n$ Taylor approximation of BaC $h$'s value after time $\delta$, if control action $u$ is taken in state $x$. The approximation is computed at the current time to predict $h$'s value after time $\delta$ and is given by: 
\begin{equation}
\hat{h}(x,u,\delta)=h(x)+\sum_{i=1}^n \frac{h^{i}(x,u)}{i!} \delta^{i}
\label{eqn:taylor}
\end{equation}
where $h^{i}(x,u)$ denotes the $i^{\rm th}$-time derivative of $h$ evaluated in state $x$ if control action $u$ is taken.  The control action is needed to calculate the time derivatives of $h$ from the definition of $h$ and Eq.~\ref{eqn:dynamics} by applying the chain rule.  
Since we are usually interested in predicting the value one time step $\eta$ in the future, we use $\hat{h}(x,u)$ as shorthand for
$\hat{h}(x,u,\eta)$.  
By Taylor's theorem with the Lagrange form of the remainder, the remainder error of the approximation $\hat{h}(x,u)$ is:
\begin{equation}
    \frac{h^{n+1}(x,u,\delta)}{(n+1)!} \eta^{n+1}\ \mbox{for some}\ \delta \in (0,\eta)
\end{equation}
An upper bound on the remainder error, if the state remains in the admissible region during the time interval, is:
\begin{equation}
\label{eqn:remainderError}
    \lambda(u)=\sup\left\{ \frac{|h^{n+1}(x,u)|}{(n+1)!} \eta^{n+1} : x \in \admis \right\}
\end{equation}

The FSC is based on checking recoverability during the next time step. For this purpose, the set $\admis$ of admissible states is shrunk by margins of $\mu_{\rm dec}$ and $\mu_{\rm inc}$, a vector of upper bounds on the amount by which each state variable can decrease and increase, respectively, in one time step, maximized over all admissible states.  Formally,
\begin{equation}
\begin{aligned}
\label{eqn:muOpt}
\mu_{\rm dec}(u)=|\min(0, \eta \dot{x}_{min}(u))|\\
\mu_{\rm inc}(u)=|\max(0, \eta \dot{x}_{max}(u))|
\end{aligned}
\end{equation}
where $\dot{x}_{min}$ and $\dot{x}_{max}$ are vectors of solutions to the optimization problems:
\begin{equation}
\begin{aligned}
\label{eqn:mu1}
\dot{x}_i^{min}(u) = \inf \{ \dot{x}_i(x,u): x \in \admis \}\\
\dot{x}_i^{max}(u) = \sup \{ \dot{x}_i(x,u): x \in \admis \}
\end{aligned}
\end{equation}
The difficulty of finding these extremal values depends on the complexity of the functions $\dot{x}_i(x,u)$.  For example, it is relatively easy if they are convex.  In both case studies described in Section~\ref{sec:EE_BCM}, they are multivariate polynomials of degree~1, and hence convex.
The set $\admis_r$ of \emph{restricted admissible states} is given by:
\begin{equation}
\label{eq:res_admis_reg}
\admis_r(u)=\{ x : x_{lb}+\mu_{\rm dec}(u) < x < x_{ub}-\mu_{\rm inc}(u) \}
\end{equation}

Let $\textit{Reach}_{=\eta}(x,u)$ denote the set of states reachable from state $x$ after exactly time $\eta$ if control action $u$ is taken in state $x$.  Let $\textit{Reach}_{\le \eta}(x,u)$ denote the set of states reachable from $x$ within time $\eta$ if control action $u$ is taken in state $x$.
\begin{lemma}
\label{lemma:admis}
For all $x \in \admis_r(u)$ and all control actions $u$, $\textit{Reach}_{\le \eta}(x,u) \subseteq \admis$.
\end{lemma}
\begin{proof}
The derivative of $x$ is bounded by $\dot{x}_{min}(u)$ and $\dot{x}_{max}(u)$ for all states in $\admis$.  This implies that $\mu_{\rm dec}$ and $\mu_{\rm inc}$ are the largest amounts by which the state $x$ can decrease and increase, respectively, during time $\eta$ as long as $x$ remains within $\admis$ during the time step. Since $\admis_r(u)$ is obtained by shrinking $\admis$ by 
$\mu_{\rm dec}$ and $\mu_{\rm inc}$ (i.e., by moving the lower and upper bounds, respectively, of each variable inwards by those amounts), the state cannot move outside of $\admis$ during time $\eta$. 
\end{proof}

\subsection{Forward Switching Condition}
\label{sec:fsc_dm}

To ensure safety, a Forward Switching Condition (FSC) should switch control from the NC to the BC if using the control action $u$ proposed by the NC causes any unsafe states to be reachable from the current state $x$ during the next control period, 
or causes any unrecoverable states to be reachable at the end of the next control period.
These two conditions are captured in the following  definition:
\begin{definition}%
\label{def:FSC}
A condition $FSC(x,u)$ is a {\em forward switching condition} if for every recoverable state $x$, every control action $u$, and control period $\eta$, $\textit{Reach}_{\leq \eta}(x,u) \cap \mathcal{U} \neq \emptyset \lor \textit{Reach}_{=\eta}(x,u) \not\subset \recov$ implies $FSC(x,u)$ is true.
\end{definition}

\begin{theorem}
A Simplex architecture whose forward switching condition satisfies Definition~\ref{def:FSC} keeps the system invariably safe provided the system starts in a recoverable state.
\end{theorem}
\begin{proof}
Our definition of an FSC is based directly on the switching logic in Algorithm~1 of \cite{yang17simplex}. The proof of Theorem~1 in~\cite{yang17simplex} shows that an FSC that is exactly the disjunction of the two conditions in our definition invariantly ensures system safety.  It is easy to see that any weaker FSC also ensures safety. \shortonly{\hfill$\blacksquare$}
\end{proof}

We now propose a new and general procedure for constructing a switching condition from a BaC and prove its correctness.
\begin{theorem}
\label{thm:FSC}
Given a continuous-time system and a barrier certificate $h$ that verifies the safety of
the system under control of the BC, the following condition is a forward switching condition: 
\begin{equation}
\label{eq:FSC}
FSC(x,u) = \alpha  \lor  \beta   
\end{equation}

\noindent
where $\alpha \equiv \hat{h}(x,u) +\lambda(u)  \geq 0$ and
$\beta \equiv x \notin \admis_r(u)$

\end{theorem} 
\begin{proof}
Intuitively, $\alpha \lor \beta$ is an FSC because (1)~if condition $\alpha$ is false, then control action $u$ does not lead to an unsafe or unrecoverable state during the next control period, provided the state remains admissible during that period; and (2)~if condition $\beta$ is false, then the state will remain admissible during that period.  Thus, if $\alpha$ and $\beta$ are both false, then nothing bad can happen during the control period, and there is no need to switch to the BC.

Formally, suppose $x$ is a recoverable state, $u$ is a control action, and 
${\textit{Reach}_{\leq \eta}(x,u) \cap \mathcal{U} \neq \emptyset} \; \lor \; \textit{Reach}_{=\eta}(x,u) \not\subset \recov$, i.e., there is an unsafe state in $\textit{Reach}_{\leq \eta}(x,u)$ or an unrecoverable state in $\textit{Reach}_{=\eta}(x,u)$.  Let $x'$ denote that unsafe or unrecoverable state.  Recall that  $\zsls(h) \subseteq \recov$, and $\recov \cap \mathcal{U} = \emptyset$.  Therefore, $h(x',u) \le 0$.  We need to show that $\alpha \lor \beta$ holds.  We do a case analysis based on whether $x$ is in $\admis_r(u)$.

Case 1: $x \in \admis_r(u)$.
In this case, we use a lower bound on the value of the BaC $h$ to show that states reachable in the next control period are safe and recoverable.  Using Lemma~\ref{lemma:admis}, we have $\textit{Reach}_{\le \eta}(x,u)  \subseteq  \admis$. This implies that $\lambda(u)$, whose value, by definition, is maximized over $x\in\admis$, is an upper bound on the error in the Taylor approximation $\hat h(x,u,\delta)$ for $\delta \le \eta$.  This implies that $\hat{h}(x,u)-\lambda(u)$ is a lower bound on value of BaC for all states in $\textit{Reach}_{\le\eta}(x,u)$. As shown above, there is a state $x'$ in $\textit{Reach}_{\le\eta}(x,u)$ with $h(x',u) \le 0$.  $\hat{h}(x,u)-\lambda(u)$ is lower bound on $h(x',u)$ and hence must also be less than or equal to 0. Thus, $\alpha$ holds.

Case 2: $x \notin \admis_r(u)$. In this case, $\beta$ holds.  Note that in this case, the truth value of $\alpha$ is not significant (and not relevant, since $FSC(x,u)$ holds regardless), because the state might not remain admissible during the next control period.  Hence, the error bound obtained using  Eq.~\ref{eqn:remainderError} is not applicable. \shortonly{\hfill$\blacksquare$}
\end{proof}

\subsection{Reverse Switching Condition}
\label{sec:rsc_dm}
In the traditional Simplex approach, after the DM switches control to the BC, the BC remains in control forever.
Moreover, there are no techniques to correct the AC after it generates an unrecoverable control input. NSA~\cite{dung2019_nsa} addresses both of these limitations.  In terms of switching control back to the AC,  it introduced two approaches to the design of Reverse Switching Conditions (RSCs).  The first approach switches control back to NC if the FSC does not hold within a specified time horizon $T$ in a simulation of the upcoming behavior. The second approach switches control back to NC if the current plant state is farther than a specified threshold distance from the nearest state for which the FSC holds. Both  approaches avoid excessive switching between the NC and BC.  Note that heuristic approaches can safely be used for the design of the RSC, since the RSC does not affect safety of the system.  In terms of correcting the AC after it generates an unrecoverable control input, NSA uses online retraining of the AC (which is a neural controller in NSA) for this purpose. 

We design the RSC to hold if the FSC is likely to remain false for at least $m$ time steps, with $m>1$.  The RSC, like the FSC, is the disjunction of two conditions.
The first condition is $h(x) \ge m \eta |\dot{h}(x)|$, since $h$ is likely to remain non-negative for at least $m$ time steps if its current value is at least that duration times its rate of change.  The second condition ensures that the state will remain admissible for $m$ time steps.  In particular,
we take:
\begin{equation}
\label{eq:RSC}
RSC(x) = h(x) \geq m\eta|\dot{h}(x)| \,\land\, x \in \admis_{r,m},
\end{equation}
where the $m$-times-restricted admissible region is:
\begin{equation}
\label{eqn:admis_rm}
\admis_{r,m}=\{ x : x_{lb}+m \mu_{\rm dec} < x < x_{ub}- m \mu_{\rm inc} \},
\end{equation}
where vectors $\mu_{\rm dec}$ and $\mu_{\rm inc}$ are defined in the same way as $\mu_{\rm dec}(u)$ and $\mu_{\rm inc}(u)$ in Eqs.~\ref{eqn:muOpt} and~\ref{eqn:mu1} except with optimization over all control actions $u$.
\fullonly{An RSC that guarantees absence of forward switches for at least $m$ time steps can be designed by using the maximum of $\dot h(x)$ over the admissible region; however, this conservative approach might leave the BC in control longer than desired.}

\subsection{Decision Logic}

The DM's switching logic has three inputs: the current state $x$, the control action $u$ currently proposed by the NC, and the name $c$ of the controller currently in control (as a special case, we take $c=NC$ in the first time step).  The switching logic is defined by cases as follows: $DM(x,u,c)$ returns $BC$ if $c=NC$ $\land$ $FSC(x,u)$, returns $NC$ if $c=BC$ $\land$ $RSC(x)$, and returns $c$ otherwise.

\section{Extending \name to Approximate Dynamical Systems}
\label{ch:ApprxDyn}

\newcommand{\appr}[1]{\underline{#1}}

\newcommand{\admisnew}{\mathcal{A}}
\newcommand{\admisr}{\mathcal{A}_r}

When a detailed, accurate, physics-based dynamic model is unavailable, or is too complex to be analyzed directly using the methods described in the previous section, an alternative approach is to use an approximate dynamic model as the basis for \name.  For example, an approximate dynamic model in the form of neural ODEs can be learned from traces, as described in \journal{Section}\thesis{Chapter}~\ref{ch:Experiments_BCM}.  In order to obtain safety guarantees, error bounds quantifying the accuracy of the approximate dynamic model are needed.  This section presents the modifications to \name needed to apply it to an approximate dynamic model, taking the error bounds into account.  This includes modifications to the definition of a BaC, the algorithms for learning and verification of neural BaCs, and the derivation of switching conditions from BaCs.

First, we introduce some notation.  Let the actual dynamics of the underlying system be given by~(\ref{eqn:dynamics}),
whereas the approximate dynamics are given by:
\begin{equation}
    \label{eqn:approx}
    \dot{x}=\appr{f}(x, g(x))
\end{equation}
When deriving a BaC from the approximate dynamics, we strengthen the definition to require that the function be a BaC, in the usual sense, with respect to all of the dynamics that are consistent with the approximate dynamics; i.e., it should be a BaC for all of the dynamics within the error bounds of the approximate dynamics. Similarly, the switching condition should be safe with respect to all of the dynamics within the error bounds of the approximate dynamics. Note that a switching condition derived from am approximate model will generally switch to the baseline controller further in advance of a potential safety violation than a switching condition derived from the actual dynamics, and it may result in more unnecessary switches to baseline control. %
The magnitude of these effects depends on the error bounds.

We now define the error bounds for the approximate dynamics.
\begin{definition}
\label{def:derivativeerror}
\emph{Derivative error} $\epsilon_i$ is the maximum absolute error between the $i^{th}$-degree derivatives of the actual and approximate dynamics with respect to all admissible states $x$.
\end{definition}
\begin{equation}
\label{eqn:dererror}
    \epsilon_i= \sup \{|\frac{\partial^i f(x,g(x))}{\partial x^i} - \frac{\partial^i \appr{f}(x,g(x))}{\partial x^i}| : x \in \admisnew\} 
\end{equation}
where $\epsilon_i$ is a vector and the $|\cdot |$ operator computes the component-wise absolute value of a vector.
Note that
$\epsilon_0$ is the maximum absolute error between the actual and approximate dynamics over the domain of admissible states.
\begin{equation}
    \label{eqn:epsilon}
    \epsilon_0= \sup \{|f(x,g(x)) - \appr{f}(x,g(x))| : x \in \admisnew  \} 
\end{equation}
where $\admisnew$ is the set of admissible states, as defined in \journal{Section}\thesis{Chapter} \ref{sec:DM}.

\subsection{Definition, Learning and Verification of BaCs using Approximate Dynamics}

The error in the approximate dynamics affects the learning and verification of a BaC since the definition of a BaC given in Eq.~\ref{eq:bac} depends on the actual dynamics.  This \journal{Section}\thesis{Chapter} treats the impact of the error on the definition of the BaC, learning of neural BaCs, and verification of neural BaCs.

\paragraph*{Modified Definition of BaC}
To show that the definition of the BaC needs to be modified, suppose $\appr{h}$ is a function that satisfies the usual conditions in the definition of a BaC (Definition~\ref{def:BaC}), but now with respect to the approximate dynamics $\appr{f}$:
\begin{align}
    \appr{h}(x) \leq 0, & \quad \forall x \in \mathbb{I} \\
    \appr{h}(x)>0, & \quad \forall x \in \mathcal{U} \\
\label{eq:bak3}
    \frac{\partial \appr{h}(x)}{\partial x}\appr{f}(x,g(x))  \leq 0, & \quad \forall x~\text{s.t.}~ \appr{h}(x)=0
\end{align}
This function is not a BaC for the actual dynamics $f$ as it does not necessarily satisfy the final condition Eq.~\ref{eq:bak3}.
Therefore, we define a function $\appr{h}$ to be a BaC with respect to approximate dynamics $\appr{f}$ as follows:
\begin{definition}
\label{def:appr_bac}
A function $\appr{h}(x)$ is a barrier certificate for a system having approximate dynamics $\appr{f}$, with maximum absolute error $\epsilon_0$ between the actual and approximate dynamics, if it satisfies the following conditions:
\begin{align}
    \label{eqn:safecon} \appr{h}(x) \leq 0, & \quad \forall x \in \mathbb{I} \\
    \label{eqn:usafecon} \appr{h}(x) > 0, & \quad \forall x \in \mathcal{U} \\
\label{eqn:apprxDerCon} y\frac{\partial \appr{h}(x)}{\partial x} \leq 0 ,  & \quad 
\artfmsdonly{
\forall x: \appr{h}(x)=0,
\forall y \in [\appr{f}(x,g(x))-\epsilon_0 , \appr{f}(x,g(x))+\epsilon_0]}
\ieeeonly{
\begin{array}[t]{@{}l@{}}
\forall x: \appr{h}(x)=0,\\
\forall y \in [\appr{f}(x,g(x))-\epsilon_0 , \appr{f}(x,g(x))+\epsilon_0]
\end{array}}
\end{align}
\end{definition}
\begin{theorem}
\label{thm:approxCond}
If a function $\appr{h}:\mathbb{R}^d \longrightarrow \mathbb{R}$ satisfies Definition~\ref{def:appr_bac}, then $\appr{h}$ is a BaC for actual dynamics $f$.
\end{theorem}
\begin{proof}
A BaC for dynamics $f$ should be non-positive over all states in the initial region, and positive over the unsafe region. As seen from Eqs.~\ref{eqn:safecon}--\ref{eqn:usafecon}, $\appr{h}$ satisfies these conditions. Also, the BaC's time derivative should be non-increasing on the boundary. We will now prove that $\appr{h}$ also satisfies this condition.

The condition in Eq.~\ref{eqn:apprxDerCon} ensures that when the approximate dynamics are used, the derivative of $\appr{h}$ is non-increasing for all possible values that $f$ can take. Since approximate dynamics have $\epsilon$-uncertainty, that condition requires that the time derivative of $\appr{h}$ is non-increasing in a box  centered at $\appr{f}(x,g(x))$; i.e., $ y \in [\appr{f}(x,g(x))-\epsilon_0 , \appr{f}(x,g(x))+\epsilon_0 ]$. I.e., it requires that the third condition in the usual definition of a BaC holds for all dynamics that are within $\epsilon_0$ of $\appr{f}$. For all states in the admissible region $\admisnew$, the definition of $\epsilon_0$ implies  $\appr{f}(x,g(x))-\epsilon_0 \leq f(x,g(x)) \leq \appr{f}(x,g(x))+\epsilon_0$. Thus, $\appr{h}$ satisfies the defining conditions of a BaC for dynamics $f$.
\end{proof}

\journal{
The training and verification for the first two conditions of the BaC remain the same as with the actual dynamics. The loss functions for these two conditions are as follows~\cite{Zhao2020}:
$l_1=\sum_{x \in D_1}-min(N(x),0), l_2=\sum_{x \in D_2} max(N(x),0)$.
We modify the loss function corresponding to the third condition in the definition of a BaC. 
We define function $w(x,\epsilon_0)$ as follows:
\begin{equation}    w(x,\epsilon_0)=\frac{\partial N(x)}{\partial x}\appr{f}(x,g(x)) + |\frac{\partial N(x)}{\partial x}|\epsilon_0
\end{equation}
Given a state $x$, $w(x,\epsilon_0)$ captures the maximum possible value of the Lie derivative of $N$ for dynamics in the $\epsilon$-box around $\appr{f}(x,g(x))$. 
The original loss function (Eq.~15) given in~\cite{Zhao2020}   is: $l_3=\sum_{x \in D_3}max(\frac{\partial N}{\partial x} f(x,g(x)),0)$.
The new loss function $l_3$ corresponding to the third condition in the definition of a BaC using approximate dynamics $\appr{f}$ is given by:
\begin{equation}
    \label{eqn:newloss}
    l_3=\sum_{x \in D_3} max(w(x,\epsilon_0),0) 
\end{equation}
Thus, if the Lie derivative of $N$ is positive for any dynamics within the $\epsilon$-box around $\appr{f}(x)$, it results in a loss.
}

For the verification of the Neural BaC, we modify the optimization problem in Eq.~38 of~\cite{Zhao2020} by replacing its original objective $v f(x,g(x))$ by $v y$, where $y$ is a new variable introduced to range over the values of $f(x,g(x))$ consistent with the approximate dynamics. %
The new optimization problem is as follows:
\begin{equation}
\begin{aligned}
   \label{eqn:approxpwlenc}
    p^*=\max \quad &y \, v \\
    \text{s.t.} \quad & x_0 \in D_3 \\
       \quad & v = \frac{\partial h}{ \partial x}\\
       \quad & y \geq \appr{f}(x, g(x))-\epsilon_0 \\
       \quad & y \leq \appr{f}(x,g(x))+\epsilon_0 \\
       \quad & x_1 = ReLU(W_1 \, x_0 + b_1)\\
     \quad & x_2 = ReLU(W_2 \, x_1 + b_2) \\
     \quad & x_3 = W_3 \, x_2 + b_3\\
     \quad & W_3 \, x_2 + b_3 = 0
\end{aligned}
\end{equation}
We have added a constraint on $y$ such that it lies in the interval $[\appr{f}(x,g(x))-\epsilon_0, \appr{f}(x,g(x))+\epsilon_0]$ instead of using the single value $\appr{f}(x,g(x))$. The verification of the remaining conditions is the same as given in~\cite{Zhao2020}. The derivative condition is verified if the optimum returned for this problem is non-positive. 
We have the following analogue of Theorem~11 in~\cite{Zhao2020}.
\begin{theorem}
\label{thm:pwlApprox}
Given a neural network $\appr{h}$ trained using approximate dynamics $\appr{f}$ with error bound $\epsilon_0$, if the optimal solution $p^*$ of the optimization problem in Eq.~\ref{eqn:approxpwlenc} is non-positive, then $\appr{h}$ is a BaC for the actual dynamics $f$.
\end{theorem}

\begin{proof}
\journal{We need to prove that neural network $\appr{h}$ satisfies Eqs.~\ref{eqn:safecon}--\ref{eqn:apprxDerCon}. We refer to~\cite{Zhao2020} for proofs of the unchanged conditions in Eqs.~\ref{eqn:safecon} and~\ref{eqn:usafecon}. For the condition in Eq.~\ref{eqn:apprxDerCon}, we modified the optimization problem in Eq.~38 of~\cite{Zhao2020} to include the range of values that $f(x,g(x))$ can assume given a state in accordance with approximate dynamics $\appr{f}(x,g(x))$ and its error $\epsilon_0$.  This ensures that the original condition holds for all dynamics within $\epsilon_0$ of $\appr{f}$, and hence it holds for the actual dynamics $f$.}

\thesis{
    According to Theorem~\ref{thm:approxCond}, we need to verify that the neural network $\appr{h}$ satisfies the conditions in Eqs.~(\ref{eqn:safecon},~\ref{eqn:apprxDerCon}). For the verification proofs of the conditions in Eqs.~(\ref{eqn:safecon},~\ref{eqn:usafecon}) we refer the reader to~\cite{Zhao2020}.
    For the verification of the condition in Eq.~\ref{eqn:apprxDerCon}, we solve the optimization problem given in Eq.~\ref{eqn:approxpwlenc}. The optimum solution of this problem $p^*$ is actually the maximum value of the product $y\frac{\partial \appr{h}}{\partial x}, \, y \in [\appr{f}(x)-\epsilon_0,\appr{f}(x)+\epsilon_0],\, x \in D_3$. If $p^*$ is non-positive over the given domains, then we can say that $y\frac{\partial \appr{h}(x)}{\partial x} \leq 0 , \, \forall x: \appr{h}(x)=0, \, \forall y \in [\appr{f}(x)-\epsilon_0, \appr{f}(x)+\epsilon_0 ]$. Thus, we have verified that the neural network $\appr{h}$ satisfies all the necessary conditions for a BaC of actual dynamics $f$.}
\end{proof}

\thesis{
In Gurobi, we first have to apply a piece-wise linear approximation to $\appr{f}$, changing the optimization problem in Eq.~\ref{eqn:approxpwlenc} to the following:
\begin{equation}
\begin{aligned}
    \label{eqn:approxpwlenc_2}
     p_r^*=\max \quad &y \, v \\
     \text{s.t.} \quad & x_0 \in D_3 \\
       \quad & v = \frac{\partial h}{ \partial x}\\
       \quad & y \geq \appr{f}_{\rm PL}(x)-\epsilon_0 \\
       \quad & y \leq \appr{f}_{\rm PL}(x)+\epsilon_0 \\
       \quad & x_1 = ReLU(W_1 \, x_0 + b_1)\\
     \quad & x_2 = ReLU(W_2 \, x_1 + b_2) \\
     \quad & x_3 = W_3 \, x_2 + b_3\\
     \quad & W_3 \, x_2 + b_3 = 0
\end{aligned}
\end{equation}
where $\appr{f}_{\rm PL}(x)$ is a PWL approximation of $\appr{f}$.
\begin{theorem}
    For the problem defined in Eq.~\ref{eqn:approxpwlenc_2}, if the solution is $p' \leq 0$, while the maximum relative error is $\xi < 1$, then the global optimal solution $p^*$ of Eq.~\ref{eqn:approxpwlenc} is non-positive.
\end{theorem}
\begin{proof}
See proof of Theorem%
\end{proof}
}

\journal{
In Gurobi, we need to replace the new inequalities in Eq.~\ref{eqn:approxpwlenc} with two PWL-equivalent inequalities $y \geq \appr{f}_{\rm PL}(x)-\epsilon_0 $ and $y \leq \appr{f}_{\rm PL}(x)+\epsilon_0$, where $\appr{f}_{\rm PL}(x)$ is a PWL approximation of $\appr{f}(x,g(x))$.
}

\subsection{Admissibility Condition in FSC using Approximate Dynamics}
\label{sec:approx-admissibility}

The admissibility condition $\alpha$ of the FSC derived using the actual dynamics checks whether the state is in a shrunken admissible region $\admisr(u)$, defined by Eq.~\ref{eq:res_admis_reg}.  
When using approximate dynamics, the restricted admissible region $\appr{\admisr}(u,\epsilon_0)$ is defined by modified versions of those equations that take the error bound $\epsilon_0$ for the dynamics into account.  This ensures that the admissibility condition holds for all actual dynamics consistent with the approximate dynamics.  Specifically,
\begin{equation}
\label{eqn:approxadmisr}
    \appr{\admisr}(u,\epsilon)=\{ x : x_{lb}+\appr{\mu}_{\rm dec}(u,\epsilon_0) < x < x_{ub}-\appr{\mu}_{\rm inc}(u,\epsilon_0) \}
\end{equation}
where vectors $\appr{\mu}_{\rm dec}(u,\epsilon_0)$ and $\appr{\mu}_{\rm inc}(u,\epsilon_0)$ are defined as follows:
\begin{equation}
\begin{aligned}
\label{eqn:muOptapprox}
\appr{\mu}_{\rm dec}(u,\epsilon) &= \eta (|\min(0, \appr{\dot{x}}_{min}(u))|+ \epsilon_0)\\
\appr{\mu}_{\rm inc}(u,\epsilon) &= \eta (|\max(0, \appr{\dot{x}}_{max}(u))|+\epsilon_0)
\end{aligned}
\end{equation}
where $\appr{\dot{x}}_{min}$ and $\appr{\dot{x}}_{max}$ are vectors of solutions to the optimization problems: \begin{equation}
\begin{aligned}
\label{eqn:mu1approx}
\appr{\dot{x}}_{min}(u) &= \inf \{ \appr{f}(x,u): x \in \admisnew \}\\
\appr{\dot{x}}_{max}(u) &= \sup \{ \appr{f}(x,u): x \in \admisnew \}
\end{aligned}
\end{equation}

The admissibility condition $\appr{\beta}$ in the FSC when using approximate dynamics with uncertainty $\epsilon_0$ is:
\begin{equation}
    \label{eqn:approximatealpha}
    \appr{\beta} \equiv x \notin \appr{\admisr}(u,\epsilon_0)
\end{equation}

\begin{lemma}
\label{lemma:maxmu}
$\appr{\mu}_{\rm dec}(u,\epsilon_0)$ and $\appr{\mu}_{\rm inc}(u,\epsilon_0)$ given by Eq.~\ref{eqn:muOptapprox} are the largest amounts by which a state $x$ can decrease and increase, respectively, during time $\eta$, provided that $x$ remains in admissible region $\admisnew$.
\end{lemma}
\begin{proof}
    The time derivative of state $x$ is bounded by $\dot{x}_{min}(u)$ and $\dot{x}_{max}(u)$. During time $\eta$, state $x$ will decrease and increase by at most $\eta|\dot{x}_{min}(u)|$ and $\eta|\dot{x}_{max}(u)|$, i.e., $\mu_{\rm dec}(u)$ and $\mu_{\rm inc}(u)$, respectively. We can obtain upper bounds on these amounts by using approximate dynamics $\appr{f}$ and error bound $\epsilon_0$. If $x$ remains within admissible region $\admisnew$, the absolute difference between $f$ and $\appr{f}$ is at most $\epsilon_0$. Hence, $\mu_{\rm dec}(u)$ and $\mu_{\rm inc}(u)$ are bounded from above by $\eta (|\min(0,\appr{\dot{x}}_{min}(u))|+ \epsilon_0)$ and $\eta (|\max(0, \appr{\dot{x}}_{max}(u))|+ \epsilon_0)$, respectively. Hence, during period $\eta$, $x$ can decrease by the maximum amount $\appr{\mu}_{\rm dec}(u,\epsilon_0)$ and increase by the maximum amount $\appr{\mu}_{\rm inc}(u,\epsilon_0)$ as long as it remains in $\admisnew$.
\end{proof}

\noindent The following analogue of Lemma \ref{lemma:admis} holds.  \journal{The proofs are similar.}

\smallskip
\begin{lemma}
\label{lemma:approxAlpha}
Let $\appr{f}$ be an approximate dynamics with error bound $\epsilon_0$.  For all $x \in \appr{\admisr}(u,\epsilon_0)$ and all control actions $u$, $\textit{Reach}_{\le \eta}(x,u) \subseteq \admis$.
\end{lemma}
\journal{
\begin{proof}
According to Lemma~\ref{lemma:maxmu}, $\appr{\mu}_{\rm dec}(u,\epsilon_0)$ and $\appr{\mu}_{\rm inc}(u,\epsilon_0)$ are the largest amounts by which the state $x$ can decrease and increase, respectively, during time $\eta$, as long as $x$ remains within $\admisnew$. Since $\appr{\admisnew}_r(u,\epsilon_0)$ is obtained by shrinking $\admisnew$ by $\appr{\mu}_{\rm dec}(u,\epsilon_0)$ and $\appr{\mu}_{\rm inc}(u,\epsilon_0)$, i.e., by moving the lower and upper bounds of each variable inwards by those amounts, the state cannot move outside of $\admisnew_r$ during $\eta$. 
\end{proof}
}

\subsection{Safety and Recoverability Conditions in FSC using Approximate Dynamics}

We first discuss the impact of approximate dynamics on the safety and recoverability conditions in the FSC when a degree-2 Taylor approximation is used for the BaC, and then present the generalization to Taylor approximations of any degree. 

Consider the degree-2 Taylor approximation $\hat{\appr{h}}$ of $\appr{h}$, given a control action $u$:
\begin{equation}
    \hat{\appr{h}}(x,u,\eta)=\appr{h}(x,u)+\frac{d\appr{h}(x,u)}{dt} \eta + \frac{1}{2}\frac{d}{dt}(\frac{d\appr{h}(x,u)}{dt})\eta^2
\end{equation}
Since we are usually interested in predicting the value one time step in the future, we use $\hat{h}(x,u)$ as shorthand for $\hat{h}(x,u,\eta)$.  Applying the chain rule and the multiplication rule for derivatives, we can expand this expression as:
\begin{equation}
\begin{aligned}
    \hat{\appr{h}}(x,u)=\,&\appr{h}(x,u)+\frac{\partial \appr{h}(x,u)}{\partial x}\appr{f}(x,u)\eta \ieeeonly{\\&}+ 
    \frac{1}{2}(\appr{f}^{T}(x,u)\frac{\partial^2 \appr{h}(x,u)}{ \partial x^2}\appr{f}(x,u)\\&+\frac{\partial \appr{h}(x,u)}{\partial x}\frac{\partial \appr{f}(x,u)}{\partial x}\appr{f}(x,u))\eta^2
\end{aligned}
\end{equation}

The terms depend on the approximate dynamics $\appr{f}$, so we need to analyze the impact of using approximate dynamics instead of actual dynamics on their values.  Specifically, the terms depend on the derivatives of $\appr{f}$, so the error cannot easily be expressed in terms of error $\epsilon_0$ alone.  We use the $1^{st}$-order derivative error $\epsilon_1$ to bound the error in the term involving $\frac{\partial \appr{f}}{\partial x}$.  Recall that the FSC needs to check whether the BaC is possibly positive during the next time step; so we need to derive an upper bound on it.  To derive an upper bound on the Taylor approximation $\hat{\appr{h}}$ of $\appr{h}$, we introduce vector parameters $\delta_0 \in [-\epsilon_0,\epsilon_0],\delta_1 \in [-\epsilon_1,\epsilon_1]$ to range over the possible values of the actual dynamics $f$ and its spatial derivative $\frac{\partial \appr{f}}{\partial x}$ that are consistent with $\appr{f}$. Thus, an upper bound $\hat{\appr{h}}_{U}$ 
on the degree-2 Taylor approximation $\hat{\appr{h}}$ of $\appr{h}$, satisfying $\hat{\appr{h}}(x,u) \le 
 \hat{\appr{h}}_U(x,u,\epsilon_0,\epsilon_1)$, is: 
\begin{equation}
\label{eqn:hhatoptim}
\begin{aligned}
 \hat{\appr{h}}_U(x,u,\epsilon_0,\epsilon_1) =&
    \sup \{
    \appr{h}(x,u)+\frac{\partial \appr{h}(x,u)}{\partial x}(\appr{f}(x,u)+\delta_0)\eta\\
    &\artfmsdonly{\quad}+\frac{1}{2}((\appr{f}^{T}(x,u)+\delta_0^T)\frac{\partial^2 \appr{h}(x,u)}{ \partial x^2}(\appr{f}(x,u)+\delta_0)\\
    &\artfmsdonly{\quad}+\frac{\partial \appr{h}(x,u)}{\partial x}(\frac{\partial \appr{f}(x,u)}{\partial x}+\delta_1)(\appr{f}(x,u)+\delta_0))\eta^2 :\\
    &\artfmsdonly{\quad} \delta_0 \in [-\epsilon_0, \epsilon_0], \delta_1 \in  [-\epsilon_1, \epsilon_1] \} 
    \end{aligned}
\end{equation}
For specific values of $x$ and $u$, the resulting optimization problem is quadratic in $\delta_0$ and multi-linear in $\delta_0$ and $\delta_1$, and hence can be solved exactly using standard optimization algorithms.

In general, we use the following formula to calculate the $n^{th}$ derivative of $\appr{h}$ w.r.t.\ time using the $(n-1)^{st}$ derivative by applying the chain rule for derivatives.
\begin{equation}
\label{eq:recur_der_h}
\begin{aligned}
        \appr{h}^n(x,u)& = \frac{d}{dt}(\appr{h}^{n-1}(x,u))\\
        &=\frac{dx}{dt}\frac{\partial}{\partial x}(\appr{h}^{n-1}(x,u))\\
        &=\appr{f}(x,u)\frac{\partial }{\partial x}( \appr{h}^{n-1}(x,u))\\
        &=\appr{f}(x,u) \nabla_x(\appr{h}^{n-1}(x,u))
        \end{aligned}
\end{equation}
This equation represents the $n^{th}$ time derivative of $\appr{h}$ in terms of $\appr{f}$ and the $(n-1)^{st}$ time derivative of $\appr{h}$.  Iterating the chain rule, we obtain an expression for $\appr{h}^n(x,u)$ in terms of $\appr{f}$, time derivatives $\appr{f}^i$ for $i < n$, and spatial derivatives of $\bar{h}$.  We need an upper bound $\appr{h}^n_U(x,u,\epsilon_0,\epsilon_1,\ldots,\epsilon_n)$ on its value for all dynamics consistent with $\appr{f}$.  We obtain this by replacing each occurrence of $\appr{f}(x)$ with $(\appr{f}(x)+\delta_0)$ and each occurrence of $\appr{f}^i(x)$ with $(\appr{f}^i(x)+\delta_i)$, where $\delta_i \in [-\epsilon_i,\epsilon_i]$.  We substitute these expressions in the $n^{\rm th}$-degree Taylor approximation of $\appr{h}$ at time $\eta$ (cf.\ Eq.~\ref{eqn:taylor}) to obtain an upper bound on that Taylor approximation:
\begin{equation}
     \label{eqn:generalEqn}
     \begin{aligned}
      &\hat{\appr{h}}_U(x,u,\epsilon_0,\epsilon_1, \ldots ,\epsilon_n) =\\
      &\quad \sup \{ \appr{h}(x,u)
      +\sum_{i=1}^{n}\frac{\appr{h}^{i}_U(x,u,\delta_0, \delta_1,\ldots ,\delta_i)}{i!} \eta^i:\\
      &\quad\quad\quad \delta_0 \in [-\epsilon_0,\epsilon_0],
      \delta_1 \in [-\epsilon_1,\epsilon_1], \ldots, \delta_n \in [-\epsilon_n,\epsilon_n] \}
      \end{aligned}
 \end{equation}

To rigorously determine whether the BaC could be positive within time $\eta$, we also need an upper bound on the remainder error in the Taylor approximation.  By modifying
Eq.~\ref{eqn:remainderError} in 
a similar way, we obtain an upper bound $\appr{\lambda}(u, \epsilon_0, \epsilon_1, \ldots ,\epsilon_{n+1})$ on the remainder error $\lambda(u)$.  We get an expression for $\appr{h}^{n+1}$ in terms of $\appr{h}^n$ and $\appr{f}$, and recursively apply 
Eq.~\ref{eq:recur_der_h} to obtain a final expression for $\appr{h}^{n+1}$ in terms of $\appr{f}$ and its first $n+1$ spatial derivatives.  This expression uses approximate dynamics, and we obtain an upper bound on its value with respect to the actual dynamics by considering all dynamics that are consistent with the approximate dynamics and the error bounds.
Similar as above, we replace the $\appr{f}$ and $\appr{f}^i$ terms with $\appr{f}(x)+\delta_0$ and $\appr{f}^i(x)+\delta_i$ where $\delta_i \in [-\epsilon_i,\epsilon_i]$, and then maximize over the $\delta_i$ to obtain the desired upper bound. The result is:
\begin{equation}
\label{eq:approxlambda}
\begin{aligned}
\appr{\lambda}(u,\epsilon_0,\epsilon_1, \ldots ,\epsilon_{n+1})=\sup \{ & \frac{|\appr{h}^{n+1}(x,u,\delta_0,\delta_1, \ldots ,\delta_{n+1})|}{(n)!} \eta^{n}:\\
&x \in \admisnew,\delta_0 \in [-\epsilon_0,\epsilon_0],\delta_1 \in [-\epsilon_1,\epsilon_1],\\& \ldots ,\delta_{n+1} \in [-\epsilon_{n+1},\epsilon_{n+1}] \}
\end{aligned}
\end{equation}
Finally, the safety and recoverability condition $\appr{\alpha}$ in the FSC using approximate dynamics is:
 \begin{equation}
     \label{eqn:approxsigma}
     \appr{\alpha} \equiv \hat{\appr{h}}_U(x,u,\epsilon_0,\epsilon_1, \ldots ,\epsilon_{n})+\appr{\lambda}(u,\epsilon_0,\epsilon_1, \ldots ,\epsilon_{n+1}) > 0
 \end{equation}

\subsection{Extended FSC for Approximate Dynamics}

The FSC using approximate dynamics is defined analogously to the FSC using the actual dynamics in Eq.~\ref{eq:FSC}.

\begin{theorem}
\label{thm:approx-fsc}
Given a positive integer $n$, approximate dynamics $\appr{f}$ with error bounds $\epsilon_0$,$\epsilon_1, \ldots ,\epsilon_{n+1}$ with respect to actual dynamics $f$, and a BaC $\appr{h}$ as defined in Eqs.~(\ref{eqn:safecon}--\ref{eqn:apprxDerCon}),  the condition 
\begin{equation}
\label{eq:apprxFSC}
FSC(x,u,\epsilon_0,\ldots,\epsilon_{n+1})=\appr{\alpha} \lor \appr{\beta}
\end{equation}
is a forward switching condition.
\end{theorem} 
\begin{proof}
Suppose $x$ is a recoverable state, $u$ is a control action, and %
$\textit{Reach}_{\leq \eta}(x,u)$ $\cap$ $\mathcal{U}$ $\neq$ $\emptyset$ $\lor$ $\textit{Reach}_{=\eta}(x,u) \not\subset \recov$,
i.e., there is an unsafe state in reach set $\textit{Reach}_{\leq \eta}(x,u)$ or an unrecoverable state in reach set $\textit{Reach}_{=\eta}(x,u)$.  Let $x'$ denote that unsafe or unrecoverable state. Therefore, $\appr{h}(x',u) > 0$.  We need to show that $\appr{\alpha} \lor \appr{\beta}$ holds.  We do a case analysis based on whether $x$ is in $\appr{\admis}_r(u,\epsilon_0)$.

Case 1: $x \notin \appr{\admis}_r(u,\epsilon_0)$. In this case, $\appr{\beta}$ holds, so the desired disjunction holds. 

Case 2: $x \in \appr{\admis}_r(u,\epsilon_0)$.
In this case, we use an upper bound on the value of the BaC $\appr{h}$ to show that states reachable in the next control period are safe and recoverable. Theorem~\ref{thm:approxCond} proves that $\appr{h}$ satisfying the conditions in Eqs.~\ref{eqn:safecon}--\ref{eqn:apprxDerCon} is a BaC for the actual dynamics $f$.  Using Lemma~\ref{lemma:approxAlpha}, we have $\textit{Reach}_{\le \eta}(x,u)  \subseteq  \admis$. This implies that $\appr{\lambda}(u,\epsilon_0,\ldots,\epsilon_{n+1})$, whose definition given in Eq.~\ref{eq:approxlambda}, maximizes its value over $x\in\admis$, is an upper bound on the error in the Taylor approximation $\hat{\appr{h}}_U(x,u,\delta,\epsilon_0,\ldots,\epsilon_n)$ for $\delta \le \eta$. Also according to Eq.~\ref{eqn:generalEqn}, $\hat{\appr{h}}_U(x,u,\epsilon_0,\ldots,\epsilon_n)$ is a maxima over the dynamics consistent with the approximate dynamics, which includes the actual dynamics.  This implies that $\hat{\appr{h}}_U(x,u,\epsilon_0,\ldots,\epsilon_n)+\appr{\lambda}(u,\epsilon_0,\ldots,\epsilon_{n+1})$ is an upper bound on value of the BaC for all states in $\textit{Reach}_{\le\eta}(x,u)$. As shown above, there is a state $x'$ in $\textit{Reach}_{\le\eta}(x,u)$ with $\appr{h}(x',u) > 0$.  Since $\hat{\appr{h}}_U(x,u,\epsilon_0,\ldots,\epsilon_n)+\appr{\lambda}(u,\epsilon_0,\ldots,\epsilon_{n+1})$ is upper bound on $\appr{h}(x',u)$, it must also be greater than 0. Thus, $\appr{\alpha}$ holds.
\end{proof}

\subsection{Bounding derivative errors}
\label{ch:Conformance}
\newcommand{\atc}[1]{\underline{#1}}

This section describes how to measure the conformance degree (see Definition~\ref{def:cd}) of the approximate dynamic model given by Eq.~\ref{eqn:approx} with respect to the actual dynamics.  From the conformance degree, we obtain the error bounds $\epsilon_i$.  

\paragraph*{Approximate trace conformance}
\label{sec:atc}

Conformance relations are used to assess the similarity of the behavior of two systems~\cite{roehm2019}. 
Roehm et al. \cite{roehm2019} classify conformance relations into several types.
\emph{Simulation} is a conformance relation that relates states. 
\emph{Trace conformance} relates only output traces, not states, where an \emph{output trace} is a sequence of outputs generated by the system during an execution, and a \emph{trace} is a sequence of states that the system passes through during an execution.
\emph{Reachset conformance} abstracts a set of traces into a set of reachable states and relates these sets of states.
\emph{Approximate simulation} and \emph{approximate trace conformance} are approximate versions of simulation and trace conformance, where the states and output traces, respectively, need only be similar (approximately equal). There are also special conformance relations for hybrid systems as well as stochastic systems.   

In our experiments with RTDS, we have access to output traces but not full system states.  Therefore, we use approximate trace conformance (ATC) to assess the similarity between an analytical dynamic model, called the  \emph{approximate system}, and the RTDS model, which we refer to as the \emph{real system}.

ATC compares the traces of the two systems starting from the same state. Here we assume that the two systems are continuous, not hybrid, systems.  The following definition of ATC is adapted from~\cite{abbas2014}, which considers hybrid systems, by omitting aspects related to modes and mode changes.  
We represent traces of the systems as \textit{timed sequences}. 
\begin{definition}
\label{def:ts}
    Let $N$ be a positive integer, $T$ be a positive real, and $Y$ be a set of possible outputs of a system.  A \textbf{timed sequence} (TS) is a function $\theta: [{0,1,...,N}]  \rightarrow Y \times [0,T]$ such that for all $i \in [{0,1,...,N}], \theta(i)=(y(i), t(i))$ with $t(0)=0, t(i) \leq t(i+1)$, and $y(i) \in Y$. 
\end{definition}
In Definition~\ref{def:ts}, $y(i)$ is the system's output at time $t(i)$.  The times $t(i)$ are measured relative to the time of the initial state of the TS. 

We define ATC in terms of closeness of TSs, using parameters $\tau$ and $\gamma$ that bound the differences in time and space, respectively, between corresponding outputs of the two systems.  
We now define ($\tau,\gamma$)-closeness of timed sequences.
\begin{definition}
Given a time interval $[0,T]$ and parameters $\tau,\gamma$, two TSs $\theta=(y,t)$ and $\theta'=(y',t')$ are $(\tau,\gamma)$-close, denoted $\theta \simeq_{\tau,\gamma} \theta'$, if:
\begin{itemize}
     \item for all $i \in dom(\theta)$, $t(i)\leq T$, there exists $j \in dom(\theta')$ such that $t(i)=t'(j), |t(i)-t'(j)|<\tau$ and $|y(i)-y'(j)|<\gamma$.   
    \item for all $k \in dom(\theta'),$ $t'(k) \leq T$, there exists $c \in dom(\theta)$ such that $t'(k)=t(c), |t'(k)-t(c)|<\tau$ and $|y'(k)-y(c)|<\gamma$.
\end{itemize}
\end{definition}

\begin{definition}
Approximate dynamical system $\mathcal{D}_{appr}$ is \textbf{approximate trace conformant} to real dynamical system $\mathcal{D}_{real}$ with $(\tau,\gamma)$-closeness, denoted as $\mathcal{D}_{real}\preceq_{(\tau,\gamma)}\mathcal{D}_{appr}$, if for all TS $\theta \in \mathcal{D}_{real}$, there exists a TS $\theta' \in \mathcal{D}_{appr}$ such that $\theta$ and $\theta'$ are $(\tau,\gamma)$-close.
\end{definition}

\begin{definition}
\label{def:cd}
Given a real number $\tau$, the \textbf{conformance degree} between $\mathcal{D}_{appr}$ and $\mathcal{D}_{real}$ is the smallest $\gamma$ such that the two systems are approximate trace conformant.
\end{definition}

\textbf{Computing the CD.}
Abbas et al. \cite{abbas2014} present an algorithm based on rapidly-exploring random trees (RRTs) to compute a lower bound on the CD between two models.  This bound can be regarded as an (under-)approximation of the CD. 
This algorithm requires the ability to run each model from an initial state selected by the algorithm.  This is not feasible for the real dynamic system in some cases, including the experiments described in \thesis{Chapter}\journal{Section}~\ref{sec:EE_BCM}.  We present a variant of the algorithm that does not require this capability.  
\journal{
Our modified algorithm samples a batch of TSs from the approximate and real systems.  Each batch is a set of pairs $(\theta, \theta')$ of TSs, where $\theta$ and $\theta'$ are TSs of $\mathcal{D}_{real}$ and $\mathcal{D}_{appr}$, respectively, executed starting from the same initial state and initial control input.  For a given value of $\tau$, the algorithm compares the two traces in each pair and calculates the smallest $\gamma$ such that the two traces are $(\tau,\gamma)$-close, and then takes the maximum of the $\gamma$'s.  This process is repeated for additional batches of TSs until the fractional increase in $\gamma$ is less than a given threshold $r$.  The algorithm returns the final value of $\gamma$ as an estimate of the CD for the two systems.  Note that
the CD is calculated separately for each state variable.
}

\paragraph*{Computing Error Bounds for Derivatives}

The definitions of approximate trace conformance and conformance degree are implicitly parameterized by the notion of the system's output used when generating TSs.  We obtain a lower bound on each derivative error $\epsilon_i$, defined in Eq.~\ref{eqn:dererror}, by applying the above algorithm to TSs generated using $\frac{\partial^i f}{\partial x^i}$ as the system's output.  In cases where the dynamics $f$ is not available in an analytical form that can be differentiated, the partial derivatives can be estimated from short sequences of consecutive states in the TS using the finite difference formulas in~\cite[Section 5.1]{numAnalysis}.

If the traces used to compute the conformance degree achieve good coverage of the systems' behaviors, and a suitably small value is used for the algorithm termination parameter $r$, then these lower bounds are expected to be reasonably tight and serve as good estimates of the derivative errors.   
We rely on these estimates because it is difficult to obtain an upper bound on $\epsilon_i$ when the model of the real system is a black box.  \thesis{For example, the method in \cite{abbas2014} for obtaining an upper bound on the CD applies when the real and approximate system models are both white-box switched linear systems.}

\section{Extending \name to Hybrid Systems}
\label{sec:hybrid}

\label{ch:hybrid}
\newcommand{\radmis}{\mathcal{A}_r}
A hybrid system is a system with multiple discrete modes, each with a different dynamics (for the continuous variables).  
Microgrids generally have multiple modes of operation and hence can be modeled as hybrid systems.  Often different modes correspond to different controllers.  For example, in a transition from grid-connected to islanded mode, a frequency-regulating controller is adopted for the DG, and the system dynamics changes to reflect the behavior of this controller.  When a Simplex architecture is used, there are different modes corresponding to the baseline controller and the advanced controller.

Safety verification for hybrid systems is generally even more challenging than for continuous systems, because the analysis needs to take all possible mode transitions into account.  This can significantly increase the number of possible behaviors that need to be considered, since mode transitions can occur non-deterministically when a guard condition is satisfied.  Thus, runtime assurance techniques such as \name are even more important for hybrid systems. We give background on hybrid systems and then extend our method for deriving switching conditions to support hybrid systems.

\subsection{Background on Hybrid Systems}

We adopt the definition of hybrid systems in the SyntheBC paper~\cite{Zhao2021}, which is based on Prajna and Jadbabaie's definition \cite{prajna2004}.

\begin{definition}
A \textbf{hybrid system} $\Pi$ is a tuple $\langle M,X,f,D,I,U,T,R \rangle$, where
\begin{itemize}
\setlength{\itemsep}{0pt}
    \item[-]$M:\{m_1,\ldots,m_k\}$ is a finite set of \emph{modes}.
    \item[-]$X$ is a set of (continuous) \emph{system variables}.
    \item[-]$f$ is a set of \emph{update functions} for the system variables $X$, with an update function $f(m, \ldots)$ for each mode $m$.  The dynamics in mode $m$ is $\dot{x}=f(m,x,g_m(x))$, where $g_m$ is the controller's control law in mode $m$.
    \item[-]$D=M \times X$ is the \emph{state space}.  For a state $(m,x)$, we refer to $x$ as the \emph{continuous state}.
    \item [-]  $I \subseteq D$ is a set of \emph{initial states}.
    \item [-]  $U \subseteq D$ is a set of \emph{unsafe states}.
    \item [-] $T : \{t(m,m') | m \neq m'\}$ is a set of sets of continuous states, where $t(m,m')$ is the set of \emph{transition states} from mode $m$ to $m'$, i.e., continuous states in mode $m$ in which a transition to mode $m'$ can occur.
    \item [-] $R : \{r_{m,m'}(x) |m \neq m'\}$ is a set of \emph{reset functions}, with a reset function for each combination of modes for which $t(m,m')$ is non-empty.  $r_{m,m'}(x)$ is the continuous state resulting from a transition to mode $m'$ from continuous state $x$ in mode $m$. 
\end{itemize}
\label{def:hybrid_sys}
\end{definition}

The projections of $I$ and $U$ on a mode $m$ are  $I(m)=\{x \,|\, (m,x) \in I\}$ and $U(m)=\{x \,|\, (m,x) \in U\}$, respectively.  Let $x_{lb}(m), x_{ub}(m) $ be operational bounds on the ranges of state variables in mode $m$, reflecting physical limits and simple safety requirements. The set $\admis(m)$ of admissible states of mode $m$ is  $\admis(m)=\{x:x_{lb}(m)\leq x \leq x_{ub}(m)\}$. The overall set of admissible states is $\admis=\bigcup_{m \in M}\admis(m)$. 

\begin{definition} 
\label{def:H_BaC} 
For a hybrid system $\Pi$ of the form in Definition \ref{def:hybrid_sys}, a set of functions $h_m$ is a barrier certificate for $\Pi$ if, for each mode $m$ in $M$, and each pair of modes $(m,m')$, the following conditions hold~\cite{Zhao2021}:
\begin{align}
\label{eq:H_init}
    & h_m(x) < 0 , \forall x \in I(m) \\
\label{eq:H_unsafe}
    & h_m(x) > 0, \forall x \in U(m)\\
\label{eq:H_der}
    &\frac{\partial h_m(x) }{\partial x}f(m,x,g_m(x))) \le 0, \forall x \mbox{\rm\ s.t. } h_m(x)=0\\
\label{eq:H_other}
    &h_{m'}(x') \leq 0, \forall x \in t(m,m') \mbox{\rm\ s.t. } h_m(x)\leq 0 \land x' = r_{m,m'}(x) 
\end{align}
\end{definition}

We use the SyntheBC method~\cite{Zhao2021} to derive and verify BaCs for each mode.

\subsection{Switching Conditions for Hybrid Systems}

Let $(\bar{m}, \bar{x})$ denote a state reachable from the current state $(m,x)$ given a control action$u$ in one control period $\eta$, i.e., $(\bar{m},\bar{x}) \in Reach_{\leq \eta}((m, x), u)$.
A forward switch is needed if $(\bar m, \bar x)$ may be unsafe or unrecoverable. We can determine this by checking whether the BaC can be positive in this state, similar as for continuous systems.  However, we need to take mode transitions into account to determine the possible modes $\bar{m}$, and use the dynamics of the relevant modes to determine which continuous states $\bar{x}$ are reachable in time $\eta$.  We modify the switching condition derivation method in \journal{Section}\thesis{Chapter}~\ref{ch:DM} to do this.

To simplify the problem, we assume that at most one mode transition can occur during one control period.  This assumption is reasonable for many systems.  In particular, it is reasonable for microgrids, because $\eta$ is small, and because (as discussed at the beginning of this section) most mode changes in microgrids are due to switching from one controller to another, and those changes  occur on control period boundaries.

With this assumption, there two cases to consider, corresponding to the two ways that an unsafe or unrecoverable state can be reached from the current state $(m,x)$: (1)~continuous evolution of the state in the current mode $m$, if no mode transition occurs during the control period, or (2)~continuous evolution in the current mode $m$, followed by a mode transition to mode $\bar m$, and then continuous evolution in mode $\bar m$.  Note that BaC condition~(\ref{eq:H_other}) ensures that the BaC $h_{\bar m}$ is negative in the state immediately after the mode transition.

The FSC is a disjunction of conditions corresponding to these two cases.  For case~(1), we can use the FSC in Eq.~\ref{eq:FSC} for continuous systems, specialized for mode $m$ (e.g., using the BaC for mode $m$) to check whether the BaC can become positive and whether an inadmissible state can be reached.   A \emph{reset state} for state $(m,x)$ under control action $u$ is a state that is reachable immediately after a mode transition (including application of the appropriate reset function) that can occur within one control period starting from state $(m,x)$ and with control action $u$ being applied. For case~(2), we first compute the set of reset states for the current state, and then check a condition similar to Eq.~\ref{eq:FSC} for each reset state.  This calculation for case~(2) leads to a more conservative FSC, because for simplicity, the scenarios can have a total duration of 2$\eta$, even though a more accurate calculation would take into account that the combined duration of the intervals before and after the mode switch is at most $\eta$.

\journal{We define the Taylor approximation of BaC $h_m$ in the same manner as its continuous counterpart in Eq.~\ref{eqn:taylor}, but with subscript $m$. Similarly, we also define $\lambda_m(u),$ $\mu_{\rm inc}(m,u)$, $\mu_{\rm dec}(m,u)$, $\dot{x}_{min}(m,u)$, $\dot{x}_{max}(m,u)$, and $\radmis(m,u)$ in a manner similar to Eqs.~\ref{eqn:remainderError},\ref{eqn:muOpt},\ref{eqn:mu1},\ref{eq:res_admis_reg} defining $\lambda(u)$, $\mu_{\rm inc}(u)$, $\mu_{\rm dec}(u)$, $\dot{x}_{min}(u)$, $\dot{x}_{max}(u)$, and $\radmis(u)$ respectively, except with an additional argument $m$ for the mode.}

\thesis{
Let $\hat{h}_m(x,u,\delta)$ denote the degree-$n$ Taylor approximation of BaC $h_m$'s value after time interval $\delta$, if control action $u$ is taken in continuous state $x$. Since we are usually interested in predicting the value one time step in the future, we use $\hat{h}_m(x,u)$ as shorthand for $\hat{h}_m(x,u,\eta)$. By Taylor's Theorem with the Lagrange form of the remainder, the remainder error of the approximation $\hat{h}_m(x,u)$ is:
}
\begin{equation}
\label{eqn:h_Taylor}
    \hat{h}_m(x,u,\delta)=h_m(x)+\sum_{i=1}^n \frac{h_m^{i}(x,u)}{i!} \delta^{i}
\end{equation}
\thesis{
An upper bound on the remainder error, if the state remains in the admissible region during the time interval, is:
}
\begin{equation}
\label{eqn:h_remainderError}
    \lambda_m(u)=\sup\left\{ \frac{|h_m^{n+1}(x,u)|}{(n+1)!} \eta^{n+1} : x \in \admis(m) \right\}
\end{equation}
\thesis{
The definition of restricted admissible region for mode $m$ is similar to Eq.~\ref{eq:res_admis_reg}, except modified to use the dynamics of mode $m$.  Given a control action $u$, a state $x$ can increase or decrease by at most $\mu_{\rm inc}(m,u)$ or $\mu_{\rm dec}(m,u)$, respectively, during one control period spent in mode $m$, where:
}
\begin{equation}
\begin{aligned}
\label{eqn:h_muOpt}
\mu_{\rm dec}(m,u)=|\min(0, \eta \dot{x}_{min}(m,u))|\\
\mu_{\rm inc}(m,u)=|\max(0, \eta \dot{x}_{max}(m,u))|
\end{aligned}
\end{equation}
\thesis{where $\dot{x}_{min}$ and $\dot{x}_{max}$ are vectors of solutions to the optimization problems:}
\begin{equation}
\begin{aligned}
\label{eqn:h_mu1}
\dot{x}_{min}(m,u) = \inf \{ f(m,x,u): x \in \admis(m) \}\\
\dot{x}_{max}(m,u) = \sup \{ f(m,x,u): x \in \admis(m) \}
\end{aligned}
\end{equation}
\thesis{
For any mode $m$ and control action $u$, the set $\radmis(m,u)$ of \emph{restricted admissible states} is created by shrinking the set $\admis(m)$ of admissible states by $\mu_{\rm inc}(m,u)$ and $\mu_{\rm dec}(m,u)$.}

\begin{equation}
\label{eqn:h_admis_r}
\radmis(m,u)=\{ x : x_{lb}(m)+\mu_{\rm dec}(m,u) < x < x_{ub}(m)-\mu_{\rm inc}(m,u) \}
\end{equation}

\begin{lemma}
\label{lemma:H_rar}
    For all $x \in \radmis(m,u)$ and all control actions $u$, if there is no mode transition during time period $\eta$, then $\textit{Reach}_{\le \eta}((m,x),u) \subseteq \admis(m)$.
\end{lemma}
\begin{proof}
Same as the proof of Lemma \ref{lemma:admis}, specialized to mode $m$.  
\end{proof}

To over-approximate the set of reset states that can occur in case~(2),
we first calculate an over-approximation $R_{nt}(m,x,u)$ of the set of reachable states from the current state $(m,x)$ with control action $u$ in one control period if no mode transition occurs (subscript ${nt}$ is mnemonic for ``no transition''). 
\begin{equation}
    \label{eqn:hybrid_resetstates}
    R_{nt}(m,x,u)=\{x': x-\mu_{\rm dec}(m,u) \leq x' \leq x+\mu_{\rm inc}(m,u)\}
\end{equation}

To obtain an over-approximation $B(m,x,u)$ of the set of reset states of $(m,x)$ under control action $u$, we find transition states $\hat x$ in $R_{nt}(m,x,u)$ and apply the reset function of the appropriate mode to each of them:
\begin{equation}
  \label{eq:resetstates}
    B(m,x,u)= \{ (\bar{m}, \bar{x})\,| \,
    \begin{array}[t]{@{}l@{}}
    \exists \bar{m} \in M,  \hat{x} \in R_{nt}(m,x,u) \cap t(m,\bar m):\ieeeonly{\\} \bar{x}=r_{m,\bar{m}}(\hat{x})\}
    \end{array}
\end{equation}

The proposed FSC is: 
\begin{equation}
    \label{eq:H_fsc}
    \exists (\bar{m},\bar{x}) \in \{(m,x)\} \cup  B(m,x,u):  
    \begin{array}[t]{@{}l@{}}
    \hat{h}_{\bar{m}}(\bar{x},u) + \lambda_{\bar{m}}(u)  > 0\ieeeonly{\\}
    \,\lor\, \bar{x} \notin \radmis(\bar{m},u) 
    \end{array}
\end{equation}
Note that the two arguments of the union correspond to the two cases described above.  Now we prove safety of this proposed FSC.
\begin{theorem}
\label{thm:H_FSC}
    Given a set of functions $h_m$, one for each mode $m$ in $M$, that form a barrier certificate, the condition in Eq.~\ref{eq:H_fsc} is a forward switching condition.
\end{theorem}
\begin{proof}
    Intuitively, the condition in Eq.~\ref{eq:H_fsc} is a disjunction of FSCs as defined in Eq.~\ref{eq:FSC} for each mode reachable from the current state in time period $\eta$. 
    Formally, let $x_{U}$ denote an unsafe or unrecoverable continuous state reachable within control period $\eta$ from the current state. Let $m'$ denote the mode of continuous state $x_{U}$.  We consider two cases depending on whether a mode transition occurs during the next control period.

\noindent
    Case 1: no mode transition occurs.     In this case, the reasoning in the proof of Theorem~\ref{thm:FSC} applies.

\noindent
    Case 2: a mode transition occurs. From the definition~(\ref{def:H_BaC}) of a BaC for a hybrid system, we know that if the BaC in the current state is non-positive, then transitions from this state to other modes are safe.  So in the case of a mode transition, the only way an unsafe continuous state $x_U$ can be reached is after a mode transition to a safe reset state. Hence in this case, we want to determine in which mode the plant ends up in after a mode transition and use the corresponding BaC and continuous state to evaluate the Taylor approximation-based FSC. Recall that $R_{nt}(m,x,u)$ is an over-approximation of the set of reachable states from the current state in one control period, i.e., $Reach_{\leq \eta}((m,x),u) \subseteq R_{nt}(m,x,u)$.  Using the reset function for every state in $R_{nt}$ and mode in transition from mode $m$, we obtain the corresponding reachable state after the mode transition in set  $B(m,x,u)$. We now have all the reachable modes from current state in  and their corresponding reachable continuous state in set $B$. Hence, if an unsafe continuous state $x_U$ is in the reachable set $\bigcup_{(\hat{m},\hat{x}) \in B}Reach_{\leq \eta}((\hat{m},\hat{x}),u)$, then the upper bound on the Taylor approximation will be positive for at least one of the states in $B$. Now we check if an unsafe continuous state can be reached for a particular mode using the upper bound on the Taylor approximation of the corresponding BaC at all reachable continuos states in that mode. The proof of correctness of this safety condition is the same as in Case~1.  Since we want all possible mode transitions to be safe, we take the disjunction of all of these conditions to get the switching condition.  
\end{proof}

The FSC in Eq.~\ref{eq:H_fsc} might be difficult or expensive to calculate.  First, we consider the difficulty of calculating $B$, and then the difficulty of calculating the FSC, given $B$.  Note that $R_{nt}$ is a box, i.e., a rectangular region.  If each $t(m,\bar m)$ is also a box, and the reset functions $r_{m, \bar m}$ are built component-wise of identity functions and constant functions (i.e., the $i$th component of $r_{m, \bar m}(x)$ is either $x_i$ or a constant), then $B(m,x,u,\eta)$ is a union of boxes, each of which can be computed in a straightforward way by intersection and projection of boxes. Note that projection corresponds to constant components of reset functions.  Our case studies satisfy these conditions.

Now consider the difficulty of calculating the FSC, given that $B$ is a union of boxes. For each disjunct in the FSC, we consider how to check whether it holds at any point in a given box.  Checking this for the second disjunct just requires intersecting two boxes and checking whether the result is empty (recall that $\radmis(\bar{m},u)$ is also a box).  Checking this for the first disjunct requires finding the maximum of $\hat{h}_{\bar{m}}(\bar{x},u)+\lambda_{\bar{m}}(u)$ when $\bar{x}$ ranges over a box, and comparing the result with zero.  For brevity, let $g_m(x) =\hat{h}_{m}(x,u)+\lambda_{m}(u)$ (we omit $u$ as an argument of $g_m$, because $u$ is constant from the perspective of this optimization problem).  The difficulty of finding $g_m$'s maximum over a box depends on its functional form (for its dependence on $x$).  In general, a non-linear optimizer might need to be used, and, depending on the functional form of $g_m$, there might be a risk of the non-linear optimizer returning a local maximum rather than the global maximum in the box.  The running time of the non-linear optimizer is also a potential concern, since it would be run  as part of the switching condition computation.  As a preliminary exploration of this, we applied MATLAB's fmincon to some representative instances of this optimization problem in which $g_m$ was a degree-4 polynomial in $x$, and the average running time was 0.27 seconds.  

Finding the maximum is easy if $g_m$'s dependence on $x$ is polynomial with degree 3 or less.  In this case, $g_m$'s derivative with respect to $x$ is quadratic, so the derivative's zeros can be found easily.  Note that the extremal values of any function over a box can occur only at (1)~points in the box where the function's derivative is zero, or (2)~the corners of the box.

We consider some specific functional forms of $g_m$ based in part on the forms that arise in our case studies.  We typically learn BaCs $h_m$ with ReLU activation functions, following~\cite{Zhao2021}; hence they are piecewise-linear functions of $x$.  If $f_m$ is also piecewise-linear, and we use a Taylor approximation of degree~3 (or less) for $\hat h_m$, then $g_m$ is a polynomial with degree~3 or less as a function of $x$, and the switching condition can be calculated relatively easily and inexpensively, without using a non-linear optimizer.  Suppose instead that $f_m$ is an approximate dynamics in the form of neural ODEs with sigmoid activation functions.  In this case, $g_m$ is also nonlinear.  While we could simply use a nonlinear optimizer in this case, more efficient approaches might also be possible, considering that the sigmoid function itself is monotonic, and its derivative decreases monotonically on each side of its single local maximum.
Exploring this in more detail is future work.

Finally, when deriving the BaC and switching condition for a controller, we do not need to consider a mode corresponding to the AC for that controller, because the BaC and switching condition depend only on the system behavior when the BC is in control.  However, if a system has multiple instances of
Simplex, in general we need to consider modes corresponding to the ACs and the BCs for the other Simplex instances, if the instances switch between controllers independently of each other.

\section{Extending \name to Hybrid Systems with Approximate Dynamics }
\label{sec:approxHybrid}
\label{ch:apprHybrid}
This \journal{Section}\thesis{Chapter} combines the extensions to hybrid systems and approximate dynamics.
Consider a hybrid system $\Pi$ as defined in Definition~\ref{def:hybrid_sys}. The approximate dynamics in mode $m$ for control action $u$ is:
\begin{equation}
    \dot{x}=\appr{f}(m,x,g_m(x)), (m,x) \in D,
\end{equation}

\begin{definition}
\label{def:h_derivativeerror}
\textbf{Derivative Error} $\epsilon_i(m)$ is the maximum absolute error between the $i^{th}$-degree derivatives of the actual and approximate dynamics with respect to states in mode $m$.
\end{definition}
\begin{equation}
\label{eqn:h_dererror}
    \epsilon_i(m)= \sup \{|\frac{\partial^i f(m,x,g_m(x))}{\partial x^i} - \frac{\partial^i \appr{f}(m,x,g_m(x))}{\partial x^i}| : x \in \admisnew(m)\} 
\end{equation}

\subsection{Impact of Approximate Hybrid Dynamics on Learning and Verification of the BaC}

We modify the definition of a BaC for hybrid systems (Definition~\ref{def:H_BaC}) to take the error bounds of approximate dynamics into account, in a manner similar to how we extended the definition of barrier certificate for continuous systems (Definition~\ref{def:BaC}) to take those error bounds into account to obtain the definition of BaC for systems with approximate dynamics (Definition~\ref{def:appr_bac}).

\begin{definition}%
A set of functions $h_m(x)$, for mode $m$ in $M$, is a BaC of a hybrid system with approximate dynamics $\appr{f}(m,x)$ if it satisfies the following conditions:
\label{def:H_appr_bac}
\begin{align}
\label{eq:A_H_init}
    & \appr{h}_m(x) < 0 , \forall x \in I(m) \\
\label{eq:A_H_unsafe}
    & \appr{h}_m(x) > 0, \forall x \in U(m)\\
\label{eq:A_H_der}
    &\frac{\partial \appr{h}_m(x) }{\partial x}y \le 0, 
\begin{array}[t]{@{}l@{}}
    \forall x \mbox{\rm\ s.t. } \appr{h}_m(x)=0,\smallskip\\
    \forall y \in [
\begin{array}[t]{@{}l@{}}
    \appr{f}(m,x,g_m(x))-\epsilon_0(m),\ieeeonly{\\}
    \appr{f}(m,x,g_m(x))+\epsilon_0(m)]
\end{array}
\end{array}\\
\label{eq:A_H_trans}
    &\appr{h}_{m'}(x') \leq 0, \forall x \in t(m,m') \mbox{\rm\ s.t. } \appr{h}_m(x)\leq 0 \land x' = r_{m,m'}(x) 
\end{align}
\end{definition}

Note that in Eq.~\ref{eq:A_H_der}, the BaC $\appr{h}_{m'}$ is BaC of mode $m'$ that uses approximate dynamics of mode $m'$, $\appr{f}(m',x',g_{m'}(x'))$. If actual dynamics are available for mode $m'$, then derivative error $\epsilon_0(m')$ will be zero and $\appr{h}_{m'}$ will have the same definition of a BaC as Definition~\ref{def:H_BaC}.

\begin{theorem}
\label{thm:h_approxCond}
Given a set of the functions $\appr{h}_m$, one for each mode $m$ in $M$, if the set of functions satisfies Definition~\ref{def:H_appr_bac}, then it is a BaC for the actual dynamics $f$.
\end{theorem}
\journal{
\begin{proof}
    The first, second, and fourth conditions in the definition of BaC for hybrid systems with exact dynamics (Definition~\ref{def:H_BaC}) follow directly from the corresponding conditions in Definition~\ref{def:H_appr_bac}.  The proof that the third condition (the derivative condition) holds is essentially the same as in the proof of the analogous theorem for continuous systems with approximate dynamics (Theorem~\ref{thm:approxCond}); since that condition applies to each mode separately, the argument trivially generalizes to this setting.
\end{proof}
}
\thesis{
\begin{proof}
A BaC for dynamics $f$ should be non-positive over all states in the safety region, it should be positive over the unsafe region, and during a mode transition, the BaC of the reset state should be non-positive. As seen from Eqs.~\ref{eq:A_H_init}, \ref{eq:A_H_unsafe}, and~\ref{eq:A_H_trans}, $\appr{h}_m$ satisfies these conditions. Along with these conditions, a BaC's time derivative should be non-increasing on the boundary. We will now prove that $\appr{h}_m$ also satisfies this condition.
The condition in Eq.~\ref{eq:A_H_der} ensures that when the approximate dynamics are used, we check that the derivative of $\appr{h}_m$ is non-increasing at all possible values that $f$ can take. Since approximate dynamics have $\epsilon_0(m)$ uncertainty, we make sure that the time derivative of $\appr{h}_m$ is non-decreasing over a box around $\appr{f}(m,x)$, i.e., $ y \in [\appr{f}(m,x)-\epsilon_0(m) , \appr{f}(m,x)+\epsilon_0(m) ]$. I.e., $\appr{h}_m$ satisfies the BaC conditions for all the dynamics that are bounded by box of size $\epsilon_0(m)$. For all states in the admissible region $\admisnew$, we have $\appr{f}(m,x)-\epsilon_0(m) \leq f(m,x) \leq \appr{f}(m,x)+\epsilon_0(m)$. Thus, the candidate BaC $\appr{h}_m$ indeed satisfies all of the necessary conditions of a BaC for dynamics $f$.
\end{proof}
}

\noindent
\paragraph*{Training and Verification of Neural BaC for Approximate Hybrid Dynamics}

To train and verify neural BaCs for hybrid systems with approximate dynamics, we apply the process described in Chapter~\ref{ch:ApprxDyn} for training and verifing BaCs for continuous systems with approximate dynamics.
We add a new variable to account for the approximate dynamics's error in both the loss function used to train a neural BaC to satisfy the derivative condition and the optimization problem used to verify that a candidate neural BaC satisfies the derivative condition, similar as in the loss function and optimization problem of derivative condition for approximate dynamics given in Eqs.~\ref{eqn:newloss} and~\ref{eqn:approxpwlenc}, respectively, except we add a new parameter for the mode $m$. The training and verification of all other conditions remain the same as in~\cite{Zhao2020}.

\subsection{Switching Conditions for Approximate Hybrid Systems}

As in \journal{Section}\thesis{Chapter} \ref{ch:hybrid}, we assume that at most one mode change can occur in one control period.  The overall design of the FSC is the same as in \journal{Section}\thesis{Chapter} \ref{ch:hybrid} for hybrid systems, with changes similar to those made in \journal{Section}\thesis{Chapter}~\ref{ch:ApprxDyn} to ensure the original FSC holds for all dynamics consistent with the approximate dynamics and error bounds.

The upper bound on the degree-$n$ Taylor approximation of BaC $\appr{h}_m$ for mode $m$, combining features of the corresponding Eq.~\ref{eqn:generalEqn} for approximate dynamics and Eq.~\ref{eqn:h_Taylor} for hybrid systems,  is given by:
\begin{equation}
     \label{eqn:A_generalEqn}
     \artfmsdonly{
     \begin{aligned}
      \hat{\appr{h}}_m^U(x,u,\epsilon_0(m),\epsilon_1(m), \ldots ,\epsilon_n(m)) = \sup \{ \appr{h}_m(x,u)+\sum_{i=1}^{n}\frac{\appr{h}_m^{i}(x,u,\delta_0, \ldots ,\delta_i)}{i!} \eta^i:\\ \delta_0 \in [-\epsilon_0(m),\epsilon_0(m)],
      \delta_1 \in [-\epsilon_1(m),\epsilon_1(m)], \ldots, \delta_n \in [-\epsilon_n(m),\epsilon_n(m)] \}
      \end{aligned}}
     \ieeeonly{
     \begin{aligned}
      &\hat{\appr{h}}_m^U(x,u,\epsilon_0(m),\epsilon_1(m), \ldots ,\epsilon_n(m)) =\\
      &\sup \{ \appr{h}_m(x,u)+\sum_{i=1}^{n}\frac{\appr{h}_m^{i}(x,u,\delta_0, \ldots ,\delta_i)}{i!} \eta^i: \delta_0 \in [-\epsilon_0(m),\epsilon_0(m)],\\
      &\quad\quad\delta_1 \in [-\epsilon_1(m),\epsilon_1(m)], \ldots, \delta_n \in [-\epsilon_n(m),\epsilon_n(m)] \}
      \end{aligned}}
\end{equation}

The remainder error for this Taylor approximation, combining features of the corresponding Eq.~\ref{eq:approxlambda} for approximate dynamics and Eq.~\ref{eqn:h_remainderError} for hybrid systems, is:
\begin{equation}
\label{eq:approxHybridlambda}
\artfmsdonly{
\begin{aligned}
\appr{\lambda}_m(u,\epsilon_0(m),\epsilon_1(m), \ldots ,\epsilon_{n+1}(m))=&\sup \{  \frac{|\appr{h}_m^{n+1}(x,u,\delta_0,\delta_1, \ldots ,\delta_{n+1})|}{(n)!} \eta^{n}:
x \in \admisnew,\\
&\artfmsdonly{\quad}\delta_0 \in [-\epsilon_0(m),\epsilon_0(m)],\delta_1 \in [-\epsilon_1(m),\epsilon_1(m)], \ldots ,\\
&\artfmsdonly{\quad}\delta_{n+1} \in [-\epsilon_{n+1}(m),\epsilon_{n+1}(m)] \}
\end{aligned}}
\ieeeonly{
\begin{aligned}
&\appr{\lambda}_m(u,\epsilon_0(m),\epsilon_1(m), \ldots ,\epsilon_{n+1}(m)) =\\
&\quad\sup \{  \frac{|\appr{h}_m^{n+1}(x,u,\delta_0,\delta_1, \ldots ,\delta_{n+1})|}{(n)!} \eta^{n}: x \in \admisnew, \\
&\quad\quad\quad\delta_0 \in [-\epsilon_0(m),\epsilon_0(m)], \delta_1 \in [-\epsilon_1(m),\epsilon_1(m)], \ldots ,\\
&\quad\quad\quad\delta_{n+1} \in [-\epsilon_{n+1}(m),\epsilon_{n+1}(m)] \}
\end{aligned}}
\end{equation}

The restricted admissible region for mode $m$, combining features of the corresponding Eq.~\ref{eqn:approxadmisr} for approximate dynamics and Eq.~\ref{eqn:h_admis_r} for hybrid systems, is:
\begin{equation}
\label{eqn:A_h_approxadmisr}
\artfmsdonly{\appr{\radmis}(m,u)=\{ x : x_{lb}(m)+\appr{\mu}_{\rm dec}(m,u,\epsilon_0) < x < x_{ub}(m)-\appr{\mu}_{\rm inc}(m,u,\epsilon_0) \}}
\ieeeonly{
\begin{aligned}
&\appr{\radmis}(m,u)=\\
&\quad\{ x : x_{lb}(m)+\appr{\mu}_{\rm dec}(m,u,\epsilon_0) < x < x_{ub}(m)-\appr{\mu}_{\rm inc}(m,u,\epsilon_0) \}
\end{aligned}}
\end{equation}
where vectors $\appr{\mu}_{\rm dec}(m,u,\epsilon_0(m)),\appr{\mu}_{\rm inc}(m,u,\epsilon_0(m))$ are obtained as follows:
\begin{equation}
\begin{aligned}
\label{eqn:A_h_muOptapprox}
\appr{\mu}_{\rm dec}(m,u,\epsilon_0) &= \eta (|\min(0, \appr{\dot{x}}_{min}(m,u))|+ \epsilon_0(m))\\
\appr{\mu}_{\rm inc}(m,u,\epsilon_0) &= \eta (|\max(0, \appr{\dot{x}}_{max}(m,u))|+\epsilon_0(m))
\end{aligned}
\end{equation}
where $\appr{\dot{x}}_{min}$ and $\appr{\dot{x}}_{max}$ are vectors of solutions to the optimization problems: \begin{equation}
\begin{aligned}
\label{eqn:mu1approxHybrid}
\appr{\dot{x}}_{min}(m,u) &= \inf \{ \appr{f}(m,x,u): x \in \admisnew(m) \}\\
\appr{\dot{x}}_{max}(m,u) &= \sup \{ \appr{f}(m,x,u): x \in \admisnew(m) \}
\end{aligned}
\end{equation}

\begin{lemma}
\label{lemma:h_maxmu}
$\appr{\mu}_{\rm dec}(m,u,\epsilon_0)$ and $\appr{\mu}_{\rm inc}(m,u,\epsilon_0)$ given by Eq.~\ref{eqn:A_h_muOptapprox} are the largest amounts by which a state $x$ can decrease and increase, respectively, during time $\eta$ provided that $x$ remains in admissible region $\admisnew(m)$.
\end{lemma}
\begin{proof}
    Same as the proof of Lemma~\ref{lemma:maxmu}, specialized to mode $m$.
\end{proof}

\begin{lemma}
\label{lemma:H_A_rar}
    Let $\appr{f}(m,x)$ be an approximate dynamics for mode $m$ with error bound $\epsilon_0(m)$. For all $x \in \appr{\radmis}(m,u)$ and all control actions $u$, $\textit{Reach}_{\le \eta}((m,x),u) \subseteq \admis(m)$ if there is no mode transition during time period $\eta$.
\end{lemma}
\begin{proof}
Same as the proof of Lemma \ref{lemma:approxAlpha}, specialized to mode $m$.  
\end{proof}

An over-approximation $\appr{R}_{nt}(m,x,u)$ of the set of reachable states from the current state $(m,x),$ with control action $u$ in one control period if no mode transition occurs, analogous to Eq.~\ref{eqn:hybrid_resetstates} for hybrid systems with exact dynamics, is given by:
\begin{equation}
    \label{appr_h_resetstates}
    \artfmsdonly{\appr{R}_{nt}(m,x,u,\epsilon_0)=\{x': x-\appr{\mu}_{\rm dec}(m,u,\epsilon_0) \leq x' \leq x+\appr{\mu}_{\rm inc}(m,u,\epsilon_0)\}}
    \ieeeonly{
    \begin{aligned}
    &\appr{R}_{nt}(m,x,u,\epsilon_0)=\\
    &\quad\{x': x-\appr{\mu}_{\rm dec}(m,u,\epsilon_0) \leq x' \leq x+\appr{\mu}_{\rm inc}(m,u,\epsilon_0)\}
    \end{aligned}}
\end{equation}
An over-approximation $B(m,x,u)$ of the set of reset states of $(m,x)$ under control action $u$, analogous to Eq.~\ref{eq:resetstates}, is given by:
\begin{equation}
    \appr{B}(m,x,u,\epsilon_0) = \bigcup_{\bar{m} \in M, \hat{x} \in \appr{R}_{nt}(m,x,u,\epsilon_0) \cap t(m,\bar m)} \{(\bar m, r_{m,\bar{m}}(\hat{x})\}
    \label{eq:resetstates_appr}
\end{equation}

The proposed FSC, analogous to Eq.~\ref{eq:H_fsc}, is:
\begin{equation}
    \label{eq:A_H_fsc}
    \artfmsdonly{
    \begin{aligned}
    \exists (\bar{m},\bar{x}) \in \{(m,x)\} \cup  \appr{B}(m,x,u,\epsilon_0):  & \hat{\appr{h}}_{\bar{m}}^U(\bar{x},u,\epsilon_0(\bar{m}),\epsilon_1(\bar{m}), \ldots ,\epsilon_n(\bar{m})) + \\ 
    &\appr{\lambda}_{\bar{m}}(u,\epsilon_0(\bar{m}),\epsilon_1(\bar{m}), \ldots ,\epsilon_{n+1}(\bar{m}))  > 0 \\
    & \lor \bar{x} \notin \appr{\radmis}(\bar{m},u,\epsilon_0)
    \end{aligned}}
    \ieeeonly{
    \begin{aligned}
    &\exists (\bar{m},\bar{x}) \in \{(m,x)\} \cup  \appr{B}(m,x,u,\epsilon_0):\\ &\quad\hat{\appr{h}}_{\bar{m}}^U(\bar{x},u,\epsilon_0(\bar{m}),\epsilon_1(\bar{m}), \ldots ,\epsilon_n(\bar{m})) + \\ 
    &\quad\appr{\lambda}_{\bar{m}}(u,\epsilon_0(\bar{m}),\epsilon_1(\bar{m}), \ldots ,\epsilon_{n+1}(\bar{m}))  > 0 \\
    &\quad\lor \bar{x} \notin \appr{\radmis}(\bar{m},u,\epsilon_0)
    \end{aligned}}
\end{equation}
\begin{theorem}
\label{thm:H_appr_FSC}
    Given a set of functions $\appr{h}_m$, one for each mode $m$ in $M$, that form a barrier certificate as defined in Definition~\ref{def:H_appr_bac} for approximate dynamics $\appr{f}(m,x)$, the condition in Eq.~\ref{eq:A_H_fsc} is a forward switching condition.
\end{theorem}
\begin{proof}
    The condition in Eq.~\ref{eq:A_H_fsc} is a disjunction of FSCs as defined in Eq.~\ref{eq:FSC} for each mode reachable from the current state in time period $\eta$. 
    Formally, let $x_{U}$ denote an unsafe or unrecoverable state reachable within control period $\eta$ from the current state. Let $m'$ denote the mode of state $x_{U}$. Then we have two cases depending on whether the mode-transition event occurs during control period $\eta$.
    
\noindent    
    Case 1: \textit{No mode transition}. 
    As there is no mode transition, we follow the proof given in Theorem~\ref{thm:approx-fsc}.

\noindent
    Case 2: \textit{Mode transition}.
    We follow the proof of case~2 in Theorem~\ref{thm:H_FSC} with two changes. First, we use the set $\appr{R}_{nt}$ instead of $R_{nt}$.  Note that $\appr{R}_{nt}$ is an over-approximation of $R_{nt}$, i.e., $R_{nt}(m,x,u) \subseteq \appr{R}_{nt}(m,x,u,\epsilon_0)$; this follows from Lemma~\ref{lemma:h_maxmu} and the definitions of these two sets in Eqs.~\ref{eqn:hybrid_resetstates} and~\ref{appr_h_resetstates}.  Second, we use the set $\appr{B}(m,x,u,\epsilon_0)$ instead of $B(m,x,u)$.  Note that $\appr{B}(m,x,u,\epsilon_0)$ is an over-approximation of set $B(m,x,u)$, i.e., $B(m,x,u) \subseteq \appr{B}(m,x,u,\epsilon_0)$; this follows from  the definitions of these sets given in  Eqs.~\ref{eq:resetstates} and~\ref{eq:resetstates_appr}.
\end{proof}

\paragraph*{Conformance of Approximate Hybrid Dynamics}

To calculate the $\epsilon_i$'s, we again use approximate trace conformance. The definition and algorithm in \journal{Section}\thesis{Chapter}~\ref{sec:atc} to estimate the conformance degree between real and approximate continuous system models generalize straightforwardly to hybrid systems by including the mode in the state.

\section{Neural Controllers}
\label{sec:NC}
\label{ch:NC}
To help address the control challenges posed by microgrids, the application of \emph{neural networks for microgrid control} is on the rise~\cite{Garcia2020}. Increasingly, reinforcement learning (RL) is being used to train powerful Deep Neural Networks (DNNs) to produce high-performance MG controllers.
We present our approach for learning neural controllers (NCs) in the form of DNNs representing deterministic control policies.  Such a DNN maps system states (or raw sensor readings) to control inputs.  The specific RL algorithm we use is Deep Deterministic Policy Gradient (DDPG)~\cite{Lillicrap2016}, with the safe learning strategy of penalizing unrecoverable actions~\cite{dung2019_nsa}. DDPG was chosen because it works with deterministic policies and is compatible with continuous action spaces.

\shortonly{
We consider a standard RL setup consisting of an agent interacting with an environment in discrete time. At each time step $t$, the agent receives a (microgrid) state $x_t$ as input, takes an action $a_t$, and receives a scalar reward $r_t$. The DDPG algorithm employs an \emph{actor-critic framework}. The actor generates a control action and the critic evaluates its quality. In order to learn from prior knowledge, DDPG uses a replay buffer to store training samples of the form $(x_t, a_t, r_t, x_{t+1})$. At every training iteration, a set of samples is randomly chosen from the replay buffer. For further details regarding the implementation of the DDPG algorithm, please refer to Algorithm~1 in~\cite{Lillicrap2016}.
}

\fullonly{
\paragraph*{Learning Neural Controllers}

The DDPG algorithm is a model-free, off-policy Reinforcement Learning method. \emph{Model-free} means that the algorithm does not have access to a model of the environment (in our case, the microgrid dynamics).  While model-free methods forego the potential gains in sample efficiency from using a model, they tend to be easier to implement and tune. An \emph{off-policy} learner learns the value of the optimal policy independently of the current learned policy.  A major challenge of learning in continuous action spaces is exploration. An advantage of off-policy algorithms such as DDPG is that the problem of exploration can be treated independently from the learning algorithm~\cite{Lillicrap2016}. Off-policy learning is advantageous in our setting because it enables the NC to be (re-)trained using actions taken by the BC rather than the NC or the learning algorithm. The benefits of off-policy retraining are further considered in Section~\ref{sec:AM}.

We consider a standard RL setup consisting of an agent interacting with an environment in discrete time. At each time step $t$, the agent receives a (microgrid) state $x_t$ as input, takes an action $a_t$, and receives a scalar reward $r_t$. The DDPG algorithm employs an \emph{actor-critic framework}. The actor generates a control action and the critic evaluates its quality. In order to learn from prior knowledge, DDPG uses a replay buffer to store training samples of the form $(x_t, a_t, r_t, x_{t+1})$. At every training iteration, a set of samples is randomly chosen from the replay buffer. For details of the DDPG algorithm, see Algorithm~1 in~\cite{Lillicrap2016}.

}

To learn an NC for DER voltage control, we designed the following reward function, which guides the actor network to learn the desired control objective. 
\begin{equation}
    \label{eq:reward}
    r(x_t,a_t) = 
    \begin{cases}
        -1000  & \text{if FSC}(x_t,a_t)\\
        \;\;\;\;\; 100  & \text{if }v_{od} \in [v_{\textit{ref}}-\epsilon, v_{\textit{ref}}+\epsilon]\\
        - w \cdot \left( v_{od} - v_{\textit{ref}} \right)^2  & \text{otherwise}\\
    \end{cases}
\end{equation}
where $w$ is a weight ($w=100$ in our experiments), $v_{od}$ is the $d$-component of the output voltage of the DER whose controller is being learned, $v_\textit{ref}$ is the reference or nominal voltage, and $\epsilon$ is the tolerance threshold. We assign a high negative reward for triggering the FSC, and a high positive reward for reaching the tolerance region ($v_{\textit{ref}} \pm \epsilon$). The third clause rewards actions that lead to a state in which the DER voltage is closer to its reference value but not in the tolerance region. Since this reward function gives a positive reward only when the voltage is in the tolerance region, the NC learns to keep the voltage in the tolerance region as much as possible.

\paragraph*{Adversarial Inputs}

Neural controllers obtained via deep RL algorithms are vulnerable to \emph{adversarial inputs}: those that lead to a state in which the NC produces an unrecoverable action, even though the NC behaves safely on very similar inputs.  NSA provides a defense against these kinds of attacks. If the NC proposes a potentially unsafe action, the BC takes over in a timely manner, thereby guaranteeing the safety of the system.  To demonstrate NSA's resilience to adversarial inputs, we use a gradient-based attack algorithm~\cite[Algorithm 4]{pattanaik2017robust} to construct such inputs, and show that the DM switches control to the BC in time to ensure safety.

Note that this algorithm does not guarantee the successful generation of adversarial inputs every time it is executed.  The success rate is inversely related to the quality of the training of the NC.  In our experiments the highest success rate for adversarial input generation that we observed is $0.008\%$.

\paragraph*{Poisoned Neural Controllers}
\label{sec:pnc}

Another potential vulnerability of systems with NCs is that the NCs might be {\em poisoned}, i.e., trained to exhibit malicious (e.g., unsafe) behavior in some special cases, while performing well most of the time (to hide their malicious nature during acceptance testing).  The training data for a poisoned controller includes some {\em malicious traces} that introduce the malicious behavior.

To generate malicious traces, we use a malicious controller that causes
safety violations.
We use a carefully selected set of malicious traces to control the conditions under which the poisoned NC will
misbehave.
Our methodology for malicious training is as follows.  First, determine from existing simulations the approximate range of values of each state variable in the initial states, and verify that the values are approximately uniformly distributed across the range.  For each variable, select a subrange whose size is a fraction $p$ of the size of the variable's entire range. %
Second, generate traces of the MG's behavior with the malicious controller, and discard a trace if any state variable's initial value lies outside the selected subrange for that variable.  The  fraction of traces that are retained (out of the total number of traces) is approximately $p^d$, where $d$ is the dimension of the input state.  We use a complementary filter for the benign (normal) traces used in the training of the poisoned controller, so that no benign training occurs in the region of the initial state space chosen for malicious training.  
An NC trained using the selected benign traces and selected malicious traces is expected to lead to a safety violation in fraction $p^d$ of the runs and perform normally in other runs. Different reward functions are used when training from malicious traces and when training from benign traces; the functions reward malicious and benign behavior, respectively.

\section{Adaptation Module}
\label{sec:AM}
The Adaptation Module (AM) performs online retraining of the NC when the NC produces an unrecoverable action that causes the DM to failover to the BC.  With retraining, the NC is less likely to repeat the same or similar mistakes in the future, allowing it to remain in control of the system more often, thereby improving performance.  We use Reinforcement Learning with the reward function defined in Eq.~\ref{eq:reward} for online retraining. 

As in initial training, we use the DDPG algorithm (with the same settings) for online retraining.
We reuse the pool of training samples (DDPG’s experience replay buffer) from initial training of the NC, by adding the training samples from online retraining to it.  As in~\cite{dung2019_nsa}, we found that this evolves the policy in a more stable fashion, as retraining samples gradually replace initial training samples in the pool.  Another benefit of reusing the initial training pool is that retraining of the NC can start almost immediately, without having to wait for enough samples to be collected online.

\shortonly{
We use off-policy retraining i.e., at every time step while the BC is active, the BC’s action is used in the training sample.  The reward for the BC's action is based on the observed next state of the system.
}

\fullonly{
We experimented with two methods of generating online retraining samples:
\begin{itemize}
    \item Off-policy retraining: At every time step while the BC is active, the BC’s action is used in the training sample.  The reward for the BC's action is based on the observed next state of the system.
    \item Shadow-mode retraining: At every time step while the BC is active, the AM takes a sample by running the NC in shadow mode to compute its proposed action, and then simulates the behavior of the system for one time step to compute a reward for it.
\end{itemize}

In our experiments, both methods produce comparable benefits.  Off-policy retraining is therefore preferable because it does not require simulation (or a dynamic model of the system) and hence is less costly.
}

\section{Experimental Evaluation}
\label{sec:EE_BCM}
\label{ch:Experiments_BCM}

We applied \name to two realistic microgrid models.  Both MGs are modeled and simulated with RTDS~\cite{rtds}, an industry-standard high-fidelity, real-time power systems simulator.  The first MG, which we refer to as the \emph{RTDS Sample MG}, is well known because it is included in the RTDS distribution.  It is described in a conference paper~\cite{MSThesis} and in more detail in an M.S.\ thesis~\cite{MSThesis2}.  We applied \name to the PV (photovoltaic) voltage controller to ensure that the PV voltage is within $\pm 5\%$ of the reference voltage.

The second MG is the \emph{Bronzeville Community Microgrid} (BCM), a real-world MG currently operational in the Bronzeville Community in Chicago, IL.  A detailed RTDS model of the MG was provided to us by Commonwealth Edison, which operates the MG.  Figure \ref{fig:bcm-diagram} shows the one-line diagram of BCM. As described in~\cite{sharma2022}, BCM consists of two independent 12 kV distribution feeders connected to the main grid via two point-of-interconnection (POI) switches POI-1 and POI-2. The top feeder, subsystem-1 (SS-1), has distributed loads, distribution automation (DA) devices, and DERs---a cluster of gas generators (GGs) totaling 4.8 MW, a 500 kW/2 MWh battery energy storage system (BESS), and a 750 kW solar
PV.
The lower feeder, subsystem-2 (SS-2), consists of loads and DA devices, but no DERs. These two subsystems can also be interconnected via the tie switches (T1, T2, or T3) under scenarios such as load transfer or full islanding. The two feeders that together form BCM provide power to a wide range of approximately 1,000 customers, including but not limited to residential as well as commercial establishments, emergency and administrative services, institutions of higher learning, and healthcare. This system has been modeled with a high granularity in RTDS. BCM's RTDS model consists 86 sections, 36 DA devices including 3 tie breakers, 6 capacitor banks, an aggregated gas generator, a BESS, and a PV system.

\begin{figure}[t]
  \centering
  \includegraphics[width=0.9\columnwidth]{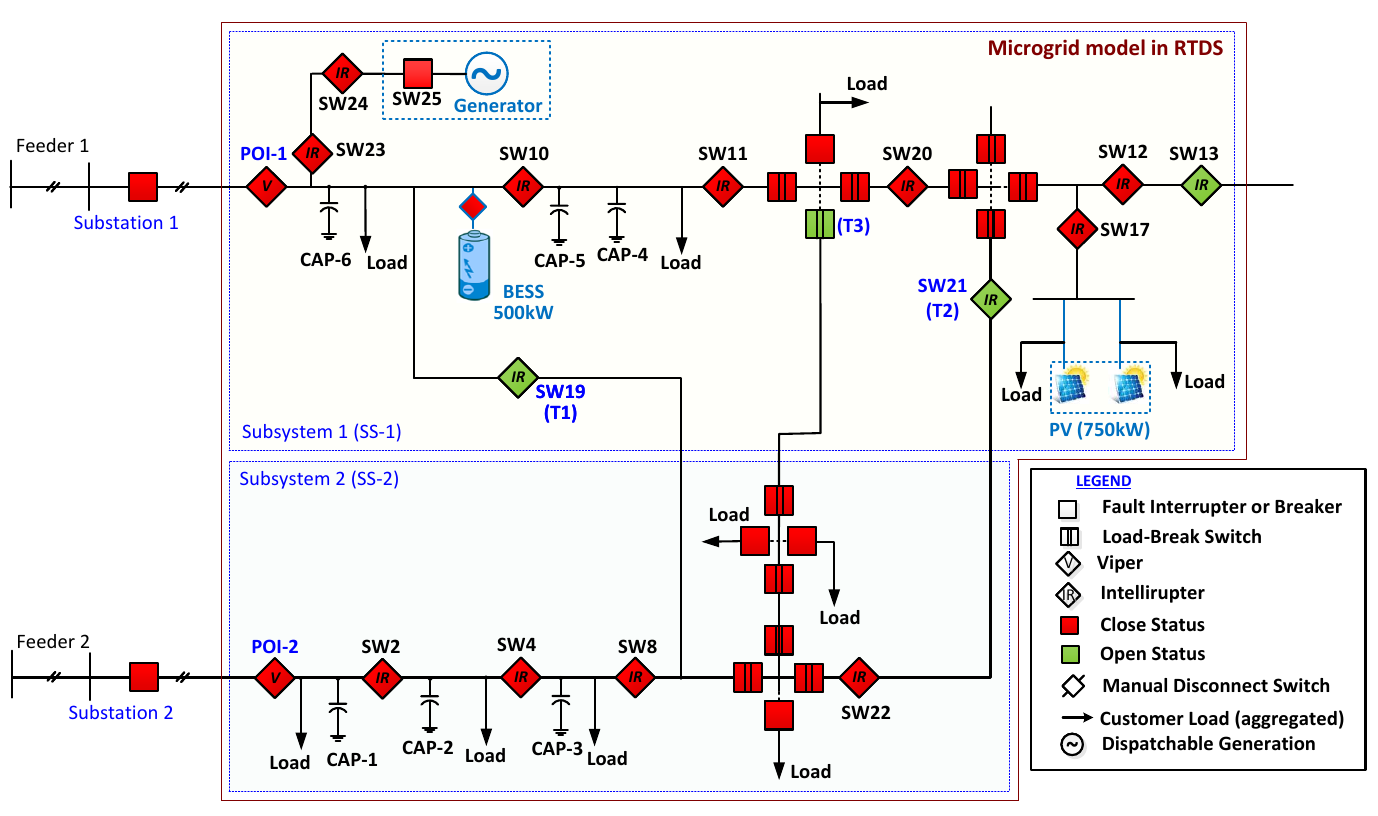}\hfill
  \caption{One-line diagram of BCM}
  \label{fig:bcm-diagram}
\end{figure}

We applied \name to the BESS voltage controller to ensure that the BESS voltage is within $\pm 5\%$ of the reference voltage $v_{\textit{ref}} = 0.48 kV$, and to the GG frequency controller to ensure that when the MG is in islanded mode, the frequency is within $\pm 5\%$ of the reference frequency $f_{\textit{ref}}=60 Hz$ (in grid-connected mode, the MG frequency matches the frequency of the main grid, and the GG frequency controller is inactive). For all experiments, we used a control time step $\eta$ of  $10^{-4}$ seconds.

Section~\ref{sec:sampleMG} summarizes our experiments with the RTDS Sample MG.  The rest of the \journal{section}\thesis{chapter}~describes our experiments with the BCM.

\thesis{
\begin{figure}[t]
  \centering
  \includegraphics[width=.8\columnwidth]{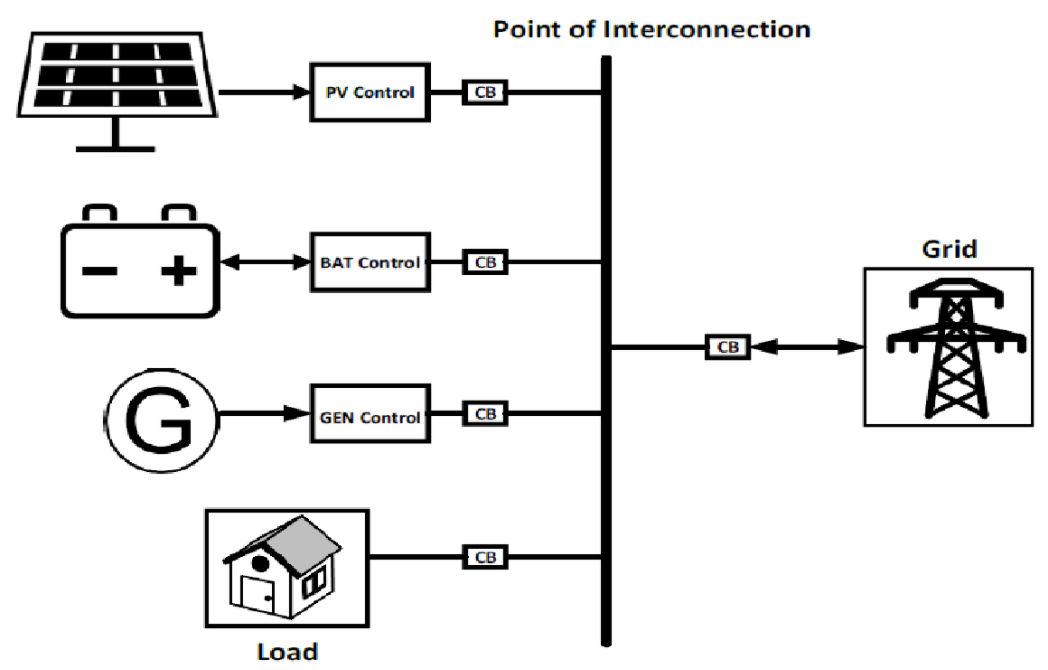}\hfill
  \caption{RTDS Sample Microgrid Model~\cite{MSThesis}}
  \label{fig:rtds_model}
\end{figure}
}
\begin{figure}[t]
\centering
\includegraphics[width=.95\columnwidth]{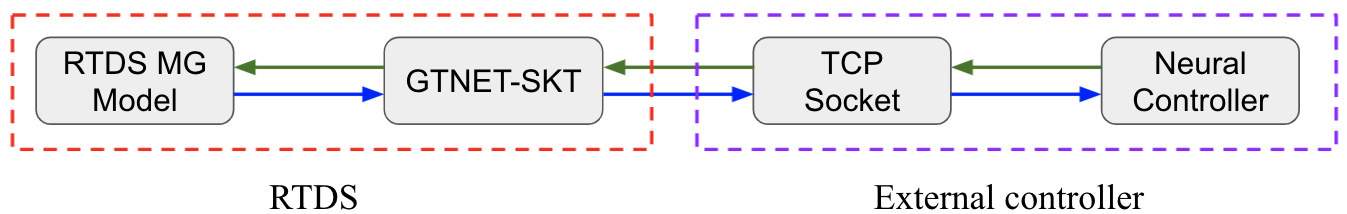}\hfill
\caption{Integration of External NC with RTDS}
\label{fig:nc_rtds}
\vspace*{-2ex}
\end{figure}

\subsection{Experiments with the RTDS Sample Migrogrid Model}
\label{sec:sampleMG}

This MG contains three DERs: a battery energy storage system (BESS), a photovoltaic DER (PV, a.k.a.\ solar panels), a diesel generator (DG), and a load.  They are connected to the main grid via bus lines and circuit breakers. \thesis{As depicted in Fig.~\ref{fig:rtds_model}, the three DERs are connected to the main grid via bus lines and circuit breakers.}  The PV control includes multiple components, such as three-phase to DQ0 voltage and current transformer, average voltage and current control, power and voltage measurements, inner-loop $dq$ current control, and outer-loop Maximum Power Point Tracking (MPPT) control.  We synthesized a BaC using the SOS-based methodology presented in~\cite{Kundu2019,permissiveBC}, using a manually developed dynamic model based on~\cite{MSThesis2}.

\thesis{We ran experiments for three configurations of the microgrid: Configuration~1: grid-connected mode with only the PV connected within the MG; Configuration~2: islanded mode with the PV and diesel generator connected within the MG; Configuration~3: islanded mode with PV, diesel generator, and battery (in discharging mode) connected within the MG.  All configurations also include a load.  We did not perform experiments with the battery in charging mode, because in this mode, the battery is simply another load, and the configuration is equivalent to Configuration~1 or Configuration~2 with a larger load.  
}

\journal{Results from our experiments with this MG appear in our prior work~\cite{Damare2022}, so we summarize them only briefly here.  We ran experiments involving several different configurations of the MG, varying in which DERs were connected to the MG, and whether the MG was in grid-connected or islanded mode.  The experiments demonstrate that (1)~the FSC is not overly conservative (it switches to BC only a few time steps before a safety violation would have occurred), (2)~the NC performs better than the BC at keeping the voltage close to the reference value, (3)~the NC generalizes well in the sense of having good performance in MG configurations not represented in its training data, (4)~the NC is susceptible to adversarial input attacks, and \name maintains safety in the face of these attacks, and (5)~the retraining of the NC by the adaptation module during the adversarial input attacks makes the NC invulnerable to
those attacks.
}

\shortonly{We ran experiments on a configuration where the PV is in islanded mode, and the diesel generator and battery (in discharging mode) DERs are connected within the MG.}

\subsection{Learning the Dynamics of BCM}
\label{sec:BCM:learn}

Manually developing an accurate analytical dynamic model of the BCM would be very difficult, because of its size and complexity.  Therefore, we learned such a model in the form of neural ordinary differential equations (ODEs)~\cite{neuralODEs}, using a PyTorch library~\cite{torchdiffeq} developed by the authors of that paper.  Specifically, we learned neural ODEs, with ReLU activation functions, that model the dynamics of the rest of the MG with respect to the BESS, i.e., the battery and its controller.
We used an analytical model of the BESS dynamics that we developed manually based on the RTDS model, because analytical dynamic models of the relevant RTDS components were available to us in~\cite{MSThesis2} and papers it references, and the effort was manageable, since the model of each DER is a small fraction of the overall size of the BCM model.  The overall dynamic model used to verify the BaC and derive the switching conditions is obtained by combining the neural ODEs and the manually-developed dynamic model of the BESS.

Furthermore, we learned two versions of the dynamic model.  First, we learned a continuous model from traces in islanded mode that contain load changes.  Second, we learned a hybrid-system model with two modes, grid-connected and islanded, from traces that contain load changes and changes between the two modes. For the continuous model, and for each mode in the hybrid-system model, we trained a neural network with two hidden layers, 32 neurons per layer, with ReLU activation functions.

For the BESS, the state space is given by $[i_{d}$, $i_q$, $i_{\textit{dref}}$, $i_{\textit{qref}}$, $v_d$, $v_q$, $v_{\textit{batd}}$, $v_{\textit{batq}}$, $I_d$, $I_q$, $V_d$, $V_q]$, where $i_d$ and $i_q$ are the $d$ and $q$ components of the BESS output current, $i_{\textit{dref}}$ and $i_{\textit{qref}}$ are the reference values for currents $i_d$ and $i_q, v_d$ and $v_q$ are the $d$ and $q$ components of the BESS output voltage, and $v_{\textit{batd}}$ and $v_{\textit{batq}}$ are the $d$ and $q$ components of the ground-truth control action used to generate the next state. $I_d$ and $I_q$ are the $d$ and $q$ components of the BESS bus current, and $V_d$ and $V_q$ are the $d$ and $q$ components of the BESS bus voltage. The dynamics governing $I_d$, $I_q$, $V_d$, and $V_q$ are neural ODEs learned from traces, whereas the dynamics governing the remaining state variables are part of the BESS dynamic model.

\paragraph*{Accuracy of the Learned Dynamic Models}
\journal{
We evaluated the accuracy of the learned dynamic models by comparing their predictions with the behavior of the RTDS model in 100 test runs from initial states not used in training and that include load changes.  For the continuous BCM model, the Mean Absolute Error (MAE) in p.u.\ and Mean Absolute Percentage Error (MAPE) for the $d$ and $q$ components of the output voltage of the BESS are 0.00165,0.0132 and 0.00157,0.25, respectively; for the current, they are 0.00157,0.983 and 0.00161,1.819, respectively.
For the hybrid-system model, the MAE in p.u.\ and MAPE for the $d$ and $q$ components of the output voltage of the BESS are 0.0009, 0.0118 and 0.0015, 0.2321, respectively; for the current, they are 0.002, 1.1889 and 0.0013, 1.5158, respectively.
}

We followed the same procedure to learn a dynamic model of the rest of the MG with respect to the GG frequency controller, and used it when verifying the BaC and deriving the switching conditions for the GG. The MAE in p.u.\ and MAPE for the $d$ and $q$ components of the output voltage of the GG are 0.0778, 0.6258 and 0.0077, 1.1711, respectively; for the current, they are 0.007968, 3.397 and 0.0007, 0.6236, respectively.

\thesis{
We trained neural ODEs for hybrid systems as well. For hybrid systems to capture the transient nature of traces after the mode change, we started our traces with a mode change from grid-connected to islanded at time~0. There was a load change at second~1 and~2. A mode switch to grid-connected occurred at second~3, followed by one more load change at second~4. Table~\thesis{~\ref{table:h_vode} summarizes our evaluation of neural ODEs for hybrid systems.
\begin{table}[t] 
\begin{subtable}{\columnwidth}
\centering
\begin{tabular}{c@{\hskip 0.2in}c@{\hskip 0.2in}c@{\hskip 0.2in}c@{\hskip 0.2in}c@{\hskip 0.2in}c}
\toprule
 Voltage component & MAE (p.u.) & MAPE (\%) \\ 
 \midrule
        d & 0.00098824 & 0.0118321  \\
        q & 0.00153786 & 0.2321271  \\
\bottomrule       
\end{tabular}
\caption{Voltage ODE}
\end{subtable}
\\[3mm]
\begin{subtable}{\columnwidth}
\centering
\begin{tabular}{c@{\hskip 0.2in}c@{\hskip 0.2in}c@{\hskip 0.2in}c@{\hskip 0.2in}c@{\hskip 0.2in}c}
\toprule
Current component & MAE (p.u.) & MAPE (\%) \\ 
\midrule
d & 0.00201657 & 1.188908 \\
        q & 0.00138877 & 1.515832 \\
\bottomrule       
\end{tabular}
\caption{Current ODE}
\end{subtable}
\caption{Accuracy of neural ODEs for BESS voltage and current for hybrid system model of BCM}
\label{table:h_vode}
\end{table}
}
}
 
\subsection{Learning BaCs and Synthesizing Switching Conditions for BCM}
\label{sec:BCM:switch}

We used the SyntheBC~\cite{Zhao2020} methodology, sketched in \journal{Section}\thesis{Chapter}~\ref{ch:BaC}, to learn and verify neural BaCs in the form of 2-layer ReLU DNNs.  The baseline controller for BESS voltage is a droop controller. The baseline controller of GG is a power controller in grid-connected mode and a frequency controller in islanded mode; both are droop controllers.  All of these are the original controllers in the BCM model.
\thesis{The dynamics of the droop-based voltage controller is:
\begin{equation}
v_{\textit{batd}}=v_d+\omega . L . i_q - k . (i_{\textit{dref}}-i_d)
\end{equation}
where $k$ is a droop gain coefficient, $\omega$ is the frequency, and $L$ is the line inductance. In our experiments, $k = 0.01$ and $L = 0.05$ p.u.}

For the BCM model, assumed to have exact dynamics, we learned neural BaCs with two hidden layers and 32 neurons per layer with ReLU activation functions. We used the same neural network architecture for all of the BaCs we learned. We experimented with different numbers of neurons and layers, and different activation functions, and choose the architecture with which we achieved the highest success rate for solving the verification optimization problem (given in Eq.~\ref{eqn:approxpwlenc}).  
 
The switching conditions for the BESS voltage and GG frequency controllers are derived from the neural BaCs using the methodology in \journal{Section}\thesis{Chapter}~\ref{ch:BaC}.  For the experiments in Section~\ref{sec:exper:continuous} with the continuous model of BCM and the experiments in Section~\ref{sec:exper:hybrid} with the hybrid model of BCM, both assumed to have exact dynamics, we used a degree-4 Taylor approximation when deriving the FSC.  For the experiments in Sections~\ref{sec:exper:approx}--\ref{sec:exper:approxhybrid} with models of BCM with approximate dynamics, we used degree-2 Taylor approximations.

A challenge in deriving the switching conditions for these models is that since the dynamics is given by neural ODEs, the optimization problems involved in calculating the restricted admissible region $\admis_r$ and the upper bound $\lambda(u)$ on the remainder error are non-convex.  
To address this issue, we used output range analysis for NNs, as implemented in the verification tool Marabou~\cite{marabou}.  First, we apply output range analysis to the neural ODEs to calculate the values of $\mu_{\rm inc}(u)$ and $\mu_{\rm dec}(u)$ over the admissible region, and use the resulting values in Eq.~\ref{eq:res_admis_reg} to compute the restricted admissible region.  Second, we get an expression for $(n+1)^{st}$ derivative of $h$ in terms of $h$ and $f$ using multiplication and chain rule of the derivative.  We then find the ranges of the terms in this expression using Marabou to do output range analysis. Finally we calculate bounds on the absolute value of the $(n+1)^{st}$ derivative of BaC $h$ by polynomial combinations of these ranges. We use this result in Eq.~\ref{eqn:remainderError} to obtain $\lambda(u)$.

\subsection{Experiments with Continuous Model of BCM}
\label{sec:exper:continuous}

We describe experiments we performed by treating the BCM model as a continuous system.  The dynamic model used to derive the BaC and FSC is the learned neural ODE model.  However, we treat it in this subsection as if it were an exact dynamic model. These experiments give us a baseline for evaluating the conservativeness of an FSC obtained using our methodology, independent of additional conservativeness due to hybrid-system mode transitions or use of approximate dynamic models.  The conservativeness of the FSC in the presence of those factors is measured in experiments in later subsections. 

\paragraph*{Evaluation of FSC}

We evaluated the conservativeness of the FSC for BESS voltage control with the MG in grid-connected mode. For this experiment, we used an AC that continuously increases the voltage of the BESS. We averaged the voltage over 100 runs from initial states with initial BESS voltages selected uniformly at random from the range $0.48$ kV $\pm \; 1\%$. The mean voltage at switching is 0.4921 kV (with standard deviation 0.0002314 kV), which is only 0.46\% below the safety threshold.  On average, the switch occurs 127.4 time steps into the trace, and a safety violation occurs 130.2 time steps into the trace if \name is not used.  Thus, on average, our FSC triggered a switch about three time steps before a safety violation would have occurred.

We performed similar experiments to determine the conservativeness of the FSC for GG frequency control, except with the MG in islanded mode.  We used an AC that continuously decreases the frequency. The mean frequency at switching is $58.8422$ Hz (with standard deviation $0.012$ Hz), which is only 0.98\% above the lower safety threshold (which is 5\% below the reference voltage).
The mean numbers of time steps before switching, and before a safety violation if \name is not used, are 128.1 and 130.2, respectively. Thus, our FSC triggered a switch about two time steps, on average, before a safety violation would have occurred.

\paragraph*{Poisoned NC Attacks}
\label{sec:exper:poisoned}

In these experiments, we demonstrate how \name protects the BCM against a poisoned NC trained using the method given in \journal{Section}\thesis{Chapter}~\ref{sec:pnc}.
\thesis{For obtaining malicious traces, we can identify the region of the initial states that will be used for generating malicious traces. The input state is $[i_d, \ i_q, \ i_{\textit{\textit{dref}}}, \ i_{\textit{qref}}, \ v_d, \ v_q]$, out of which $i_{\textit{\textit{dref}}}$ and $i_{\textit{qref}}$ are constants. Thus, we need to concern ourselves with only the remaining four variables. 
After identifying the initial ranges of the four variables, we verified that these indeed follow an uniform distribution. We chose $85\%$ of the range as a filtering for the initial states of malicious traces. Since there are $d=4$ variables, choosing $p=0.85$ will yield a poisoned neural controller that exhibits malicious behavior $(0.85)^4 = 0.522$ i.e., roughly 50\% of the all runs.
We also ensured that no benign training occurs in the region chosen as the malicious trace range, and the initial states for benign training came from ranges outside of the malicious ones.}

To train the poisoned controller, we used a malicious controller that simply takes the control action generated by the BC and adds a constant offset $R=0.5$ to it,
and we used the following reward function:
\begin{equation}
\label{reward:pnc}
r(x,a) = 
\begin{cases}
100  & \ieeeonly{\hspace{-0.7em}}\text{if FSC}(x,a)\\
-100 & \ieeeonly{\hspace{-0.7em}}\text{if } v_d = v_{\textit{ref}} \pm 1\%\\
-w \cdot \min( (v_d - v_{lb})^2, (v_d - v_{ub})^2 ) & \ieeeonly{\hspace{-0.7em}}\text{otherwise}\\
\end{cases}
\end{equation}
where $w=100$ is a weight and $v_{lb}=0.456$ kV and $v_{ub}=0.504$ kV are the lower and upper safety bounds of the BESS voltage.  This assigns a high positive reward for triggering the FSC, and a high negative reward for staying close to the reference voltage i.e., $v_{\textit{ref}} \pm 1\%$. The third clause rewards actions that lead to a state in which the DER voltage is close to the FSC.

\thesis{Ideally, this poisoned NC should lead to a safety violation during roughly half of all runs, since the malicious traces should account for $(0.85)^4$ of the total training samples. During the testing of the poisoned NC for 50 runs, we observed that 27 of them start in a region of malicious training. Out of those 27 runs, 21 of them cause a voltage violation. The remaining 6 also show bad behavior however, they do not eventually cause a violation within the runtime. The remaining 23 runs show normal performance.}\journal{ We chose a fraction $p$ of the size of the variable’s entire range (cf. \journal{Section}\thesis{Chapter} \ref{ch:NC}) so that this poisoned NC leads to a safety violation with probability approximately 0.5.} Figure~\ref{fig:pnc} demonstrates the performance of the poisoned NC as it leads to a safety violation of the BESS voltage. The vertical pink lines denote when switching occurs. The FSC is triggered at 21~sec and the RSC is triggered at 26~sec. The switching logic is not overly conservative since the forward switch occurs only 0.002 kV before the safety threshold.

\begin{figure}[t]
\begin{minipage}[b]{.49\textwidth}
\centering
\includegraphics[width=1\textwidth]{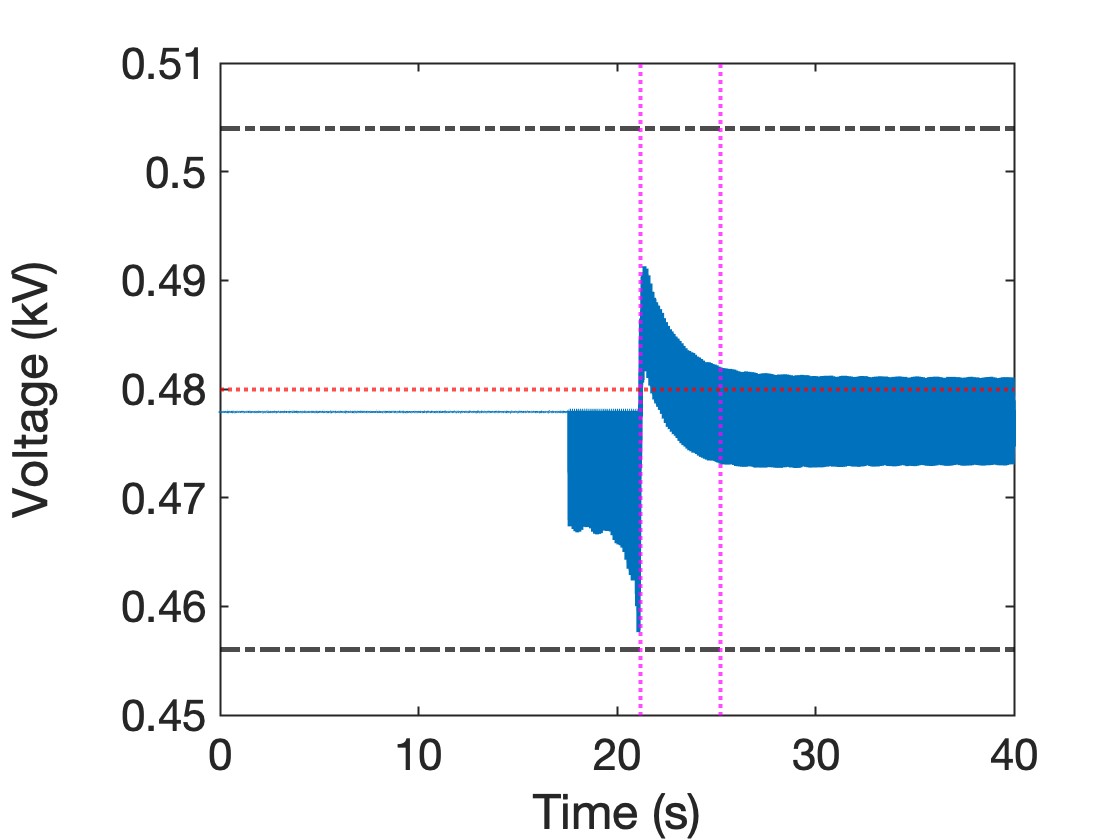}
\caption{\name with Poisoned NC}
\label{fig:pnc}
\end{minipage}
\hfill
\begin{minipage}[b]{.49\textwidth}
\centering
\includegraphics[width=1\textwidth]{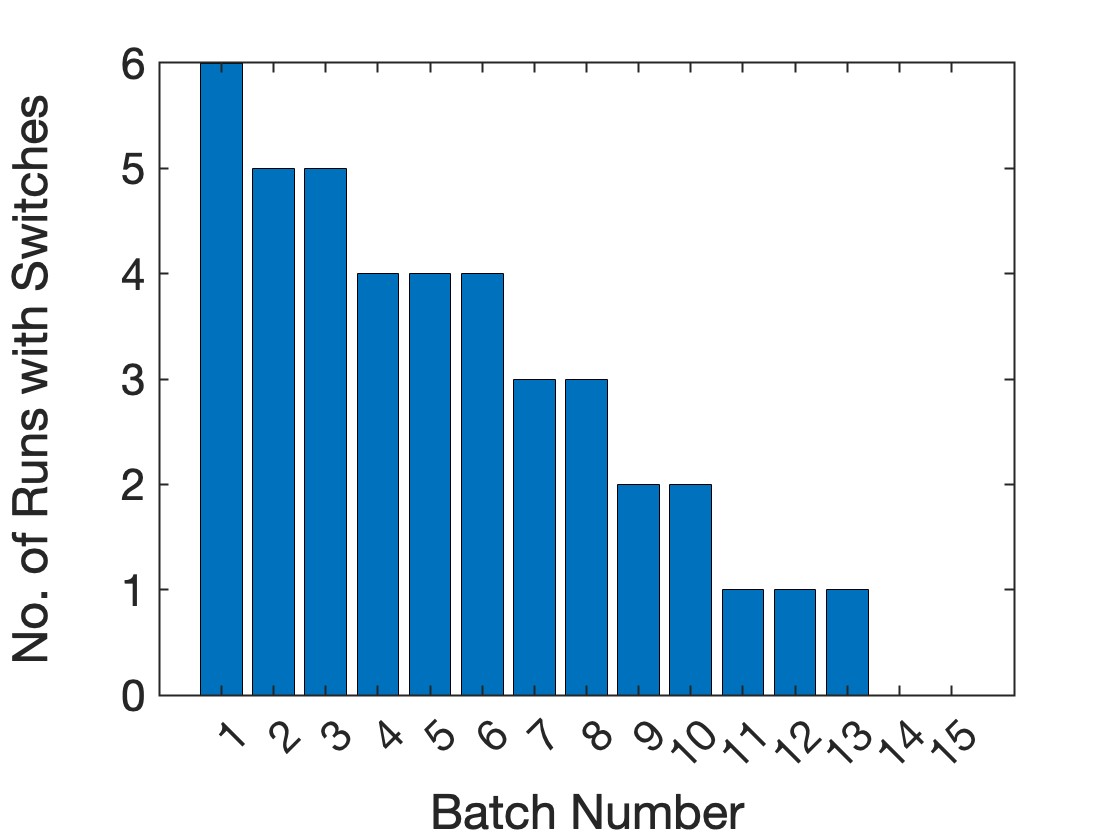}
\caption{Unpoisoning of Poisoned NC}
\label{fig:retrain_bcm}
\end{minipage}
\vspace*{-2ex}
\end{figure}

\subsection{Experiments for Continuous Model of BCM with Approximate Dynamics}
\label{sec:exper:approx}

The experiments in this section take into account that the learned neural ODE dynamics is approximate, using the method in Section~\ref{ch:ApprxDyn}.   We derived a BaC and switching conditions for the BESS voltage controller.  We estimated the bounds $\epsilon_i$ on the derivative errors using the method based on approximate trace conformance in Section~\ref{sec:atc}.
\journal{
We learn the dynamics from timed sequences (TSs) that include load changes.  For the experiments in this section, we model the approximate system as a continuous systems in grid-connected mode. The approach in Section~\ref{sec:atc} is used to estimate the error bound $\epsilon_i$.  A difficulty when applying this approach to our BCM case study is that time derivatives are not directly measurable in RTDS.  We circumvent this difficulty by approximating the time-derivatives from the observed states.  We approximate the first time-derivative of a function
using the $5$-point finite-difference formula in~\cite[Section 5.1]{numAnalysis}.
}

These experiments use the learned neural ODE model of the dynamics of the rest of the MG with respect to the BESS, as described in Section \ref{sec:exper:continuous}.  In this model, the neural ODEs give the dynamics of the $d$,$q$ components of the voltage and current.
We used a two hidden-layer neural network with 32 neurons per layer and ReLU activation functions to train neural ODEs. We obtained the following error bounds for the dynamics: 
$\epsilon_0=[0.15192274, 0.09432583, 0.00005734, 0.00135681]$, where the state components are in the order $I_d$, $I_q$, $V_d$, $V_q$. Using these bounds, we trained and verified neural BaCs and then derived the switching conditions.

\paragraph*{Evaluation of the FSC}

To evaluate the conservativeness of the FSC, we performed the same BESS voltage control experiment as described in Section~\ref{sec:exper:continuous}.  Based on 100 runs from random initial states, switching occurred 8.47 time steps on average before violation. The mean voltage at switching was $0.4909$ kV, which is only 1.0\% below the safety threshold.
As expected, the FSC derived using approximate dynamics is more conservative than the FSC derived assuming exact dynamics, for which the switch occurred on average about 3 steps before a safety violation.

\subsection{Experiments with Hybrid-System Model of BCM}
\label{sec:exper:hybrid}

We model the BCM as a hybrid system with two modes: grid-connected and islanded. We use the hybrid extension of \name to ensure that the GG frequency $f$ is always within $\pm 2 Hz$ of the reference frequency $f_{\textit{ref}} = 60 Hz$ for both modes. It is important to ensure the frequency safety because in islanded mode, GG provides the reference frequency for the rest of the grid, and an excessive increase in GG power output leads to a corresponding decrease in the frequency. 
The GG's Governor (baseline controller) controls the GG's power in grid-connected mode and its frequency in islanded mode.  In both cases, it is a droop controller.\thesis{ The Governor's control law is:
\begin{equation}
    \dot{T_m}=(1-m).k_1.(P_{\textit{ref}}-P)+m.k_2.(1-\frac{\omega}{\omega_{\textit{ref}}})
\end{equation}
where $T_m$ is the mechanical torque, $P$ is the GG active power, $P_ref$ is the GG
reference active power, $\omega$ is the machine speed, $\omega_{\textit{ref}}$ is the reference machine speed, $m$ is the mode (0 for grid-connected, 1 for islanded), and $k1, k2$ are the droop gain coefficients, with $k_1 =\frac{1}{50}$, $k_2 =\frac{1}{20}$, $P_{\textit{ref}} = 0.5$ MW, and $\omega_{\textit{ref}} = 2\pi f_{\textit{ref}}$ rad/sec and $\omega=2\pi f$.}
The experiments in this section, as in Section~\ref{sec:exper:continuous}, use the learned neural ODE dynamic model for BCM but treat it as if it were an exact dynamic model.

\paragraph*{Evaluation of the FSC}

We derive the switching conditions for each mode using the method in \journal{Section}\thesis{Chapter}~\ref{ch:hybrid}.  To evaluate the conservativeness of the FSC, we take the AC to be a dummy controller that continously increases the power in grid-connected mode and the frequency in islanded mode, in order to generate switches to BC in both modes.
We averaged over 100 runs from random initial states and over the two modes.
The mean frequency at switching is 58.8994 Hz (with a standard deviation of 0.0278 Hz), which is 0.8 Hz above the safety threshold.  
Also, our FSCs for the hybrid system triggered a forward switch about nine time steps, on average, before a safety violation would have occurred. In contrast, the FSC ($4^{th}$-order) for the continuous model of BCM triggered a forward switch about two time steps before a safety violation would have occurred in islanded mode.  This additional conservativeness of the FSC for the hybrid system reflects the difficulty of analyzing possible mode transitions in the derivation of the FSC.

\subsection{Experiments for Hybrid System Model of BCM with Approximate Dynamics}
\label{sec:exper:approxhybrid}

We applied the extension of \name~for hybrid Systems with approximate dynamics, described in \journal{Section}\thesis{Chapter}~\ref{ch:apprHybrid}, to BCM, specifically to the neural ODEs for the dynamics of the rest of the MG with respect to the GG.  We estimated bounds on the first-derivative errors in each mode using the same methodology as for a model with approximate dynamics, described in Section~\ref{sec:exper:approx}.
For grid-connected mode, $\epsilon_0=[0.00809, 0.44218, 0.13829, 0.46825]$.  For islanded mode, $\epsilon_0=[0.00912, 0.3110, 0.09743, 0.61233]$.  Using these bounds, we trained and verified neural BaCs and then derived the switching conditions.

\paragraph*{Evaluation of FSC}
We evaluated the conservativeness of the FSC in each mode.  
The FSC for the GG frequency controller in islanded mode was triggered on average 9.87 time steps before violation, and the mean frequency at switching was $61.4261$ Hz, averaged over 100 runs, which was 0.57 Hz below the safety threshold.
The FSC for the GG power controller in grid-connected mode was triggered on average 10.11 time steps before violation. The mean value of the frequency at switching was $58.8422$ Hz, which was $0.84$ Hz above the safety threshold.
In comparison with the hybrid system model with exact dynamics, where switching occurs two time steps before a violation, the approximate model demonstrates that, on average, switching happens 10 time steps prior to a violation. This quantifies the impact of approximate dynamics on the conservativeness of the switching conditions in hybrid systems.

\section{Related Work}
\label{sec:related}

\paragraph*{Related work based on the Simplex Architecture}

Yang et al.\ originated the use of BaCs in the Simplex architecture~\cite{yang17simplex}.  There are, however, significant differences between their method for obtaining the switching condition and ours.  Their switching logic involves computing, at each decision period, the set of states reachable from the current state within one control period, and then checking whether that set of states is a subset of the zero-level set of the BaC.  Our approach avoids the need for online reachability calculations by using a Taylor approximation of the BaC, and bounds on the BaC's derivatives, to bound the possible values of the BaC during the next control period and thereby determine recoverability of states reachable during that time. Our approach is computationally much cheaper: a reachability computation is expensive compared to evaluating a polynomial.  Also, their work does not consider how to handle the uncertainty of approximate dynamic models.

The Simplex methodology has been applied to hybrid systems by Yang et al.~\cite{yang17simplex}, whose approach is discussed above, and Bak et al.~\cite{bak_mitra_2010}. In~\cite{bak_mitra_2010}, a finite discrete transition system is constructed as an abstraction of an underlying hybrid dynamic system.  Reachability analysis is applied to this abstract model to derive and verify switching conditions. Like~\cite{yang17simplex}, this approach requires complex reachability calculations.

Mehmood et al.~\cite{mehmood2021simplex,mehmood2023simplex} propose a distributed Simplex architecture with BCs synthesized using control barrier functions (CBFs) and with switching conditions derived from the CBFs, which are BaCs satisfying additional constraints. A derivation of switching conditions based on a Taylor approximation of CBFs is briefly described but does not consider the remainder error, admissible states, or restricted admissible states, and the paper does not include a proof of correctness (which requires an analysis of the remainder error).

An automated method of generating safety proofs for runtime-assurance architectures is proposed in~\cite{nigam2023}. Apart from the Simplex architecture, there are several other variants of runtime-assurance architectures, e.g.~\cite{mehmood2021simplex,mehmood2023simplex,ramakrishna2020dynamic}. The primary distinction between the Simplex architecture and other runtime-assurance architectures lies in the inclusion of redundant controllers and decision logic that determines the active controller, as opposed to a decision logic solely focused on identifying unsafe behavior.~\cite{nigam2023} provides an automated proof of the decision logic, but the burden of coming up with decision logic is left to the user.  Also, their method only applies to systems with deterministic dynamics; our work allows nondeterminism through the use of approximate dynamics and nondeterministic mode changes in hybrid systems.

The work of~\cite{Zhong_2023} proposes a run-time assurance framework for neural network-based controllers using a supervisor that checks the probabilistic relation between the DNN controlled system and the safety-based controlled system. If this relation increases beyond some threshold, then it ignores the action given by neural network. The underlying system is modelled as a stochastic system and yields a probabilistic approach to safety. Compared to \name, it doesn't take recoverability into account for the safety-based controller. Their approach also requires the use of their stochastic safety controller, whereas \name allows any baseline controller to be used and derives the switching conditions based on the chosen BC.

\paragraph*{Related work on barrier certificates}

Kundu et al.~\cite{Kundu2019} and Wang et al.~\cite{permissiveBC} use BaCs to ensure safety of microgrids.
Their work allows use of verified-safe controllers only; unlike our Simplex-based approach, they do not allow the use of unverified high-performance controllers, do not consider switching conditions, etc.

Various methods have been proposed to synthesize BaCs for stochastic systems, e.g., \cite{SALAMATI2024,yang2020efficient,wang2023,Lavaei2022}.  These methods model nondeterministic system as stochastic systems and generate BaCs using either a data-driven approach~\cite{SALAMATI2024} or by modeling jumps as a stochastic difference equation with noise~\cite{Lavaei2022}. These methods provide bounds on safety probability of the system by generating BaCs. Our method bounds the error of an approximate system and then derives a BaC giving deterministic safety guarantees. 

\paragraph*{Related work on neural controllers for microgrids}

The application of neural networks for microgrid control is gaining in popularity~\cite{Garcia2020}.  Amoateng et al.~\cite{Amoateng2018} use adaptive neural networks and cooperative control theory to develop microgrid controllers for inverter-based DERs. Using Lyapunov analysis, they prove that their error-function values and weight-estimation errors are uniformly ultimately bounded.  Tan et al.~\cite{Tan2020} use Recurrent Probabilistic Wavelet Fuzzy Neural Networks (RPWFNNs) for microgrid control, since they work well under uncertainty and generalize well.  We used more traditional DNNs, since they are already high-performing, and our focus is on safety assurance, not neural control.  Our \name framework, however, allows any kind of neural network to be used as the AC and can provide the safety guarantees lacking in their work. Unlike our approach, none of these works provide safety guarantees.

\section{Conclusion}
\label{sec:conclusion}
We have presented \name, a new, provably correct design methodology for the runtime assurance of continuous dynamical systems and its extensions as outlined in Section~\ref{ch:Intro}. We also demonstrated the effectiveness of \name for runtime assurance of continuous systems and hybrid systems with exact as well as approximate dynamics through extensive experiments on a high-fidelity model of a real-life MG.
Our experiments also illustrated \name's capability to protect MGs from attacks involving poisoned neural controllers.

An important direction for future work is to extend \name to provide runtime assurance for temporal safety properties. %
Another direction for future work is to extend the SyntheBC methodology~\cite{Zhao2021}, which we are using to synthesize and verify BaCs for baseline controllers, to handle discrete-time controllers; currently, it applies only to continuous-time controllers.  This will allow the use of frameworks such as Wang et al.'s software-defined control~\cite{wang2023} to implement baseline controllers in software.

\fmsdonly{\paragraph*{Funding} \acks}
\artonly{\paragraph*{Acknowledgements} \acks}

\artonly{\bibliographystyle{alpha}}
\ieeeonly{\bibliographystyle{IEEEtran}}
\bibliography{references}
\ieeeonly{\EOD}
\end{document}